\documentclass{article}
\usepackage[utf8]{inputenc}
\usepackage{subfiles}
\usepackage[margin=1in]{geometry}
\usepackage{amsmath,amsthm,amssymb}
\usepackage{graphicx}
\usepackage{enumitem}

\newtheorem{definition}{Definition}[section]

\newtheorem{proposition}[definition]{Proposition}
\newtheorem{lemma}[definition]{Lemma}
\newtheorem{theorem}[definition]{Theorem}
\newtheorem{corollary}[definition]{Corollary}
\newtheorem{observation}[definition]{Observation}

\newcommand{\Exp}{\mbox{Exp}}

\title{Product Forms for FCFS Queueing Models with Arbitrary Server-Job Compatibilities: An Overview}
\author{Kristen Gardner and Rhonda Righter }
\date{June 2020}

\begin{document}

\maketitle

\begin{abstract}

In recent years a number of models involving different compatibilities between jobs and servers in queueing systems, or between agents and resources in matching systems, have been studied, and, under Markov assumptions and appropriate stability conditions, the stationary distributions have been shown to have product forms. We survey these results and show how, under an appropriate detailed description of the state, many are corollaries of similar results for the Order Independent Queue. We also discuss how to use the product form results to determine distributions for steady-state response times.     
    
\end{abstract}

\section{Introduction}
\label{sec:intro}

Systems in which servers are flexible in the types of customers that they
can serve, and customers are flexible in the servers at which they can be
processed, are very common in a wide range of practical settings. In call
centers, service representatives may be trained to handle different subsets
of requests, or may speak different languages. A customer who speaks only
Spanish can be helped by a representative who speaks only Spanish, or by a
representative who speaks both Spanish and English, or by a representative
who speaks both Spanish and Mandarin. \ In computer systems, some jobs may
be able to run only on those servers that have the job's data stored
locally, other jobs may require a server with a particular combination of
resources, and still other jobs may be able to run on any server. In
ride-sharing systems, drivers will only be assigned to users that are
\textquotedblleft nearby\textquotedblright\ in some sense.

This type of model is called a skill-based server model in the call center
literature. In the scheduling literature, the compatibility constraints
between job classes and servers are called eligibility constraints or
processing set restrictions, and the models are typically deterministic. For
matching models, compatibilities may be location based. While the language
and notation used to describe these models differ across research
communities, the common idea in all of the above examples is that the system
consists of multiple servers and multiple classes of jobs, with a bipartite
graph structure indicating which classes of jobs can be served by which
servers.

The examples above, and more broadly the \textquotedblleft flexible
job/server\textquotedblright\ models that exist in the literature, vary in
precisely how the bipartite matching structure is used to assign servers to
jobs. We mainly consider two service models, which we call the
\textquotedblleft collaborative\textquotedblright\ and \textquotedblleft
noncollaborative\textquotedblright\ models. In the collaborative model,
multiple servers can work together, with additive service rate, to process a
single job. This matches the computer systems setting, in which the same
(replicated) job can run on several different servers at once. In the
noncollaborative model, a customer can only enter service at a single
server. This matches the structure of a call center, in which a single
customer cannot speak with multiple representatives at the same time. In
both cases, we think of there being a single central queue for all
customers. When a server becomes available, it begins working on the next
compatible job in the queue, in first-come first-served (FCFS) order. In the
noncollaborative case, we must also specify which server will serve an
arriving job that finds multiple idle compatible servers. We will consider
two policies: Assign Longest Idle Server (ALIS), which is analogous to FCFS,
and Random Assignment to Idle Servers (RAIS).

An additional feature of many of models of service systems with job/server
compatibilities is redundancy, or job replication, i.e., the possibility of
sending multiple copies of the same job to multiple servers. For example,
this is a common practice in computer systems to combat unpredictable system
variability, so the hope is that the job may experience a significantly
shorter response time at one of the servers. Similarly, one idea for
reducing wait times on organ transplant waitlists is to allow patients to
join the waitlist in multiple geographic areas at the same time. Patients
are restricted in which waitlists they can join based on travel time: should
an organ become available at a particular hospital, the patient must be able
to travel to that hospital within a relatively short time frame to receive
the transplant. Generally systems with redundancy are not modeled as a
central FCFS queue as described above. Instead, each server has its own
dedicated queue and an arriving job can join the queues of multiple servers.
In the collaborative case, multiple copies of the same job can run on
different servers at the same time, and when the first copy completes
service all other copies are removed immediately from other servers or
queues. This is called cancel-on-completion or late cancellation. In the
noncollaborative case, all other copies of a job are removed from the system
as soon as the first copy enters service. This is called cancel-on-start or
early cancellation, and it is also equivalent to sending a single copy to
the queue with the least work. In both cases, the cancellations occur
without penalty. While the central FCFS queue and the job redundancy model
describe very different system dynamics, the two views turn out to be
sample-path equivalent, provided that service times are exponentially
distributed and i.i.d. across jobs and servers. We will explore this
relationship, as well as other model equivalences, in what follows.

Throughout most of this paper, we will make a few key assumptions: that jobs
of each class arrive according to independent Poisson processes, that
service times are exponentially distributed and i.i.d. across jobs and
servers, and that the scheduling discipline is FCFS. Under these
assumptions, we will see that the models introduced above, as well as
related models, exhibit product-form stationary distributions. Indeed,
product forms hold for several different state descriptors, each of which
provides different advantages in understanding system behavior. We first
consider the most detailed state descriptor, which tracks the classes of all
jobs in the system. This description lends itself to a concise proof of the
product form for the collaborative model, due to the Order Independence (OI)
results of Berezner and Krzesinski \cite{BK}, \cite{K}. We show that for the
noncollaborative model, both the job queue and the idle server queue are OI
queues, resulting in a product of product forms. We extend these
arguments to collaborative and noncollaborative
models with abandonments. We also show that the same product-form
stationary distribution holds for several related models, 
including new results for two-sided matching models with arrivals of both jobs and servers, and for make-to-stock inventory models with back ordering. 

Following the development of the product-form stationary distributions, we turn to using these results to derive system performance metrics. We begin with class-based response time distributions. We show that if there is a job class that is compatible with all servers, that class has an exponentially distributed response time for the collaborative model; indeed,
the response time for that class is the same as it would be for the M/M/1
queue in which all jobs are fully flexible. For the noncollaborative model,
the queueing time for that fully flexible class is a mixture of a mass at 0
and an exponential random variable. We use this result to show response time
distributions for all job classes in the collaborative model, and queueing
time distributions for all classes in the noncollaborative model, when the
compatibility matching has a nested structure.

Product-form distributions for an alternative, partially aggregated state
descriptor have been derived in the literature; we show that these results also follow as corollaries to the detailed product forms. 
The partially aggregated state description allows us to derive per-class response time distributions, conditioned on the set of busy servers and the order of the jobs they are currently serving.

We briefly discuss a related queueing model in which the state description
is the number of jobs of each class in the system (per-class
aggregation). While this state space no longer yields a Markovian description of the systme evolution, it has the same steady-state per-class mean
performance measures (mean number in system, probability the system is
empty) as the collaborative system. This state descriptor yields a simple, recursive approach to derive the system load and mean response time for our models when they are not nested.

We note that the product forms discussed in this paper are not the same as those obtained in the well-known Jackson and Kelly networks~\cite{J,Kelly}.
The standard Jackson and Kelly product forms arise in networks of queues, where the state of the network can be expressed as a product of the states at each queue.
In contrast, in this paper we primarily concentrate on
the internal product-form structure of the steady-state distribution for a
single queue.
While most of our focus is on single nodes with flexible jobs and servers, these nodes are quasi-reversible under our modeling assumptions, so a
network of such nodes also has a product form stationary distribution. 
That is, in steady state the distributions of each node will be as if they were
operating independently, as is the case in Jackson and Kelly networks.

Throughout the paper we provide pointers to the relevant literature in
context.

\section{Model}
\label{sec:model}

\begin{figure}
	\centering
	\includegraphics[scale=0.5]{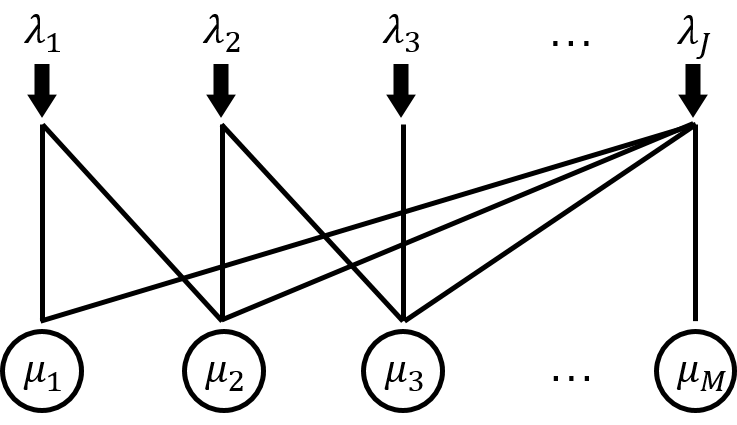}
	\caption{The system consists of $J$ classes of jobs, $M$ servers, and a bipartite matching structure indicating which job classes can be served by which servers.}
	\label{fig:bipartite}
\end{figure}

We note at the outset that our analysis requires a heavy dose of notation that we will often reuse and abuse in the interest of readability and ease of understanding. Notation that we use throughout the paper is summarized in Table~\ref{tab:notation}.

There are $J$ job classes with Poisson arrivals at rates $\lambda_{i}$, $M$ parallel servers with exponential service rates $\mu _{m}$, and a bipartite graph matching structure indicating which servers can serve which job classes (see Figure~\ref{fig:bipartite}).
For job class $i$, let $S_{i}=\{j:$ server $j$ can serve class $i\}$, and for a subset of job classes, $A$, let $S(A)= \bigcup\limits_{i\in A}S_{i}$ be the set of servers that can serve those classes. 
For example, for the system shown in Figure~\ref{fig:bipartite}, $S_1 = \{1,2\}$ and for $A = \{1,2\}$, $S(A) = \{1,2,3\}$.
For server $j$, let $C_{j} = \{ i:$ server $j$ can serve class $i \}$ be the set of job classes it can serve, and for a subset of servers, $B$, let $C(B)=$ $\bigcup\limits_{j\in B}C_{j}$ be the set of job classes that can be served by servers in $B$. 
For example, in Figure~\ref{fig:bipartite}, $C_3 = \{2,3,J\}$ and for $B = \{1,3\}$, $C(B) = \{1,2,3,J\}$.
For a subset of job classes, $A$, let $\mu (A)=\sum_{m\in S(A)}\mu _{m}$ and $\lambda(A)=\sum_{i\in A}\lambda _{i}$ be the total service rate and arrival rate for job classes in $A$, and, abusing notation, for a subset of servers, $B$, let $\mu (B)=\sum_{m\in B}\mu _{m}$ and $\lambda (B)=\sum_{i\in C(B)}\lambda_{i}$ be the total service rate and arrival rate for servers in $B$. 
It will be clear from the context whether the arguments of $\lambda $ and $\mu $ are job classes or servers. 
Finally, let $\mu =\sum_{m=1}^{M}\mu _{m}$ and $\lambda =\sum_{i=1}^{J}\lambda _{i}$ be the total system service rate and total system arrival rate respectively. 

Throughout, we will assume for stability that $\lambda (A) < \mu (A)$ for all subsets of job classes $A$.
We note that this condition is both necessary and sufficient for stability in the model described above (in particular, given i.i.d.\ exponential service times); absent this modeling assumption stability is a much more complicated question. (see Section~\ref{sec:related_work} for a more detailed discussion).

\begin{figure}
	\centering
	\includegraphics[scale=0.45]{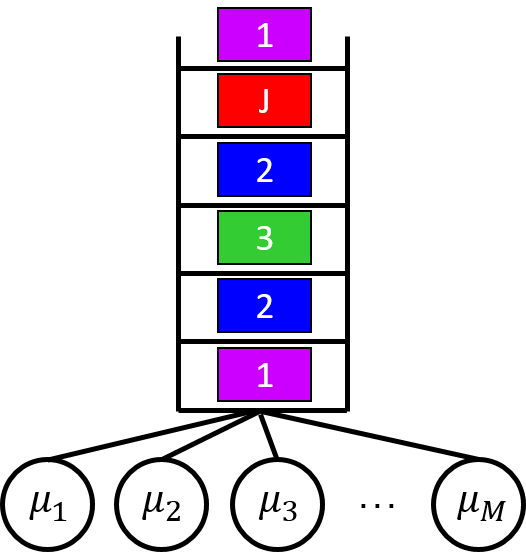} \hspace{0.3in}
	\includegraphics[scale=0.45]{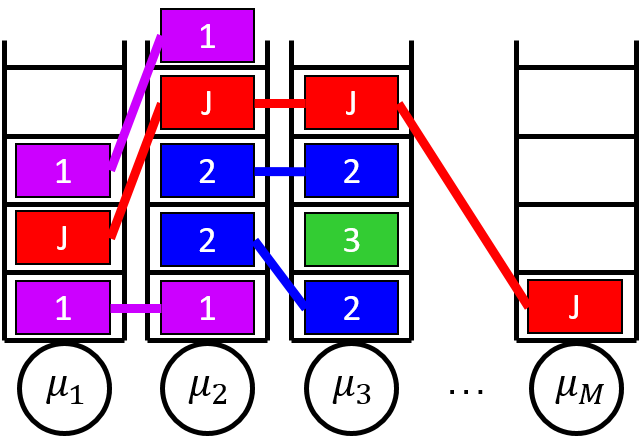} \hspace{0.3in}
	\includegraphics[scale=0.45]{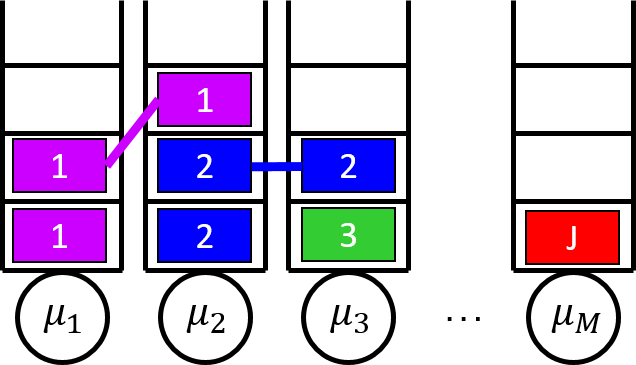} \\
	(a) Central FCFS queue \hspace{0.6in} (b) Collaborative model \hspace{0.6in} (c) Noncollaborative model
	\caption{The system can be viewed, equivalently, as (a) having a single FCFS queue, or as being a distributed system in (b) the collaborative model or (c) the noncollaborative model.}
	\label{fig:system_equivalences}
\end{figure}

We primarily consider two models of service: the collaborative model and the noncollaborative model.
In the noncollaborative model, a job can only be served by a single server.
When the first copy of a job enters service, all other copies are removed from the system immediately without penalty.
A job that arrives to the system and finds multiple idle compatible servers begins service on one of those servers, chosen according to some assignment rule.
We consider two assignment rules.
Under Assign Longest Idle Server (ALIS), the arriving job begins service on the compatible server that has been idle for the longest time.
Under Random Assignment to Idle Servers (RAIS), the arriving job chooses an idle server randomly; this selection must be drawn from a particular distribution, which we discuss in more detail in Section~\ref{sec:rais}.

In the collaborative model, a job may be in service at multiple servers at the same time.
When the first copy of a job completes service, all other copies are removed from the system immediately without penalty.
A job that is in service at a set of servers $B$ receives combined service rate $\mu(B)$, hence the job experiences an exponential service time with rate $\mu(B)$.
Unlike in the noncollaborative case, no assignment rule is needed for an arriving job that finds multiple idle compatible servers; such a job simply enters service on all of the idle compatible servers.
Note that in the collaborative case (but not in the noncollaborative case), we can assume without loss of generality that the set of job classes a server can serve is unique to that server, i.e., $C_i \neq C_j$ for $i\neq j$. 
This follows because of the FCFS and collaborative assumptions; if $C_i = C_j $ for $i\neq j$, then servers $i$ and $j$ will always be serving the same (oldest compatible) job, so they can be considered to be a single server with rate $\mu _{i}+\mu _{j}$. 
In the noncollaborative model we allow multiple servers that are identical
in their service rates and their sets of compatible job classes. 

There are two equivalent ways of viewing the system dynamics.
In the first, shown in Figure~\ref{fig:system_equivalences}(a), all arriving jobs join a single FCFS queue.
When a server $j$ becomes available, it begins working on the first job in the central queue that has class $i \in C_j$.
In the collaborative model, the ``queue'' contains all jobs in the system, including those currently in service, so that a newly available server may begin working on a job that is already in service at some other server.
In the noncollaborative model, the queue contains only those jobs that are not in service.
The second system view is that of a distributed system, in which each server has its own queue and works on the jobs in its queue in FCFS order.
Here, an arriving job of class-$i$ joins the queue at all servers in $S_i$.
In the collaborative model (Figure~\ref{fig:system_equivalences}(b)), multiple copies of the job may be in service at different servers.
For example, in Figure~\ref{fig:system_equivalences}(b) the class-1 job shown at the head of the queue at both servers 1 and 2 is in service at both servers.
In the noncollaborative model (Figure~\ref{fig:system_equivalences}(c)), only one copy of a job can be in service.
In our example, the class-1 job shown at the head of the queue at server 1 is in service at server 1, and its other copy has been removed from the queue at server 2.
Another equivalent model in the noncollaborative case is to assume, again, that each server has its own dedicated queue, and that each arriving job in class-$i$ is routed to the server in $S_i$ with the least work (i.e., Join-the-Shortest-Work among compatible servers)~\cite{ABV d,ABV}.


Throughout the remainder of this paper, we will rely primarily on the central-queue view of the system when developing our state descriptors.
We introduce here the notation used in the state descriptors.
This notation captures a great deal of information about the system, and each state descriptor uses a slightly different subset of this information to capture different aspects of the system dynamics.
We elaborate further on the specific state descriptors in the sections that follow.
Let $\vec{c}_n = (c_1,\dots,c_n)$ denote the classes of all jobs in the central queue, where $c_i$ is the class of the $i$th job in the queue in order of arrival (so $c_1$ is the class of the oldest job, and $c_n$ is the class of the most recent arrival).
As noted above, for the collaborative model the ``queue'' refers to all jobs in the system, including those in service, whereas for the noncollaborative model the ``queue'' refers to only those jobs that are not in service.
Let $\vec{b}_l = (b_1,\dots,b_l)$ be the vector of busy servers in the arrival order of the jobs that they are serving (so $b_1$ is serving the oldest job in the system, and $b_l$ is serving the most recent arrival among the jobs in service).
We use $\vec{z}_m$ to denote, in the noncollaborative model, an interleaving of $\vec{c}_n$ and $\vec{b}_l$ ordered by job arrival time, where the state tracks the job class for positions corresponding to jobs in the queue, and it tracks the busy server for positions corresponding to jobs in service.
Let $\vec{s}_k = (s_1,\dots,s_k)$ be a vector of idle servers in the order in which they became idle, where $l + k = M$.
We use $\vec{d}_{l} = (d_1,\dots,d_{l})$ to denote the classes of the jobs currently in service, where $d_i$ is the class of the job in service at server $b_i$.
The vector $\vec{n}_{l} = (n_1,\dots,n_{l})$ denotes the number of jobs waiting to be served ``between'' the jobs in service.
That is, $n_i$ gives the number of jobs that arrived after the job in service at server $b_i$ and before the job in service at server $b_{i+1}$.
Finally, $\vec{x}_J = (x_1,\dots,x_J)$ denotes the number of jobs of each class in the system; $x_i$ is the number of class-$i$ jobs in the system.

\begin{table}
	\centering
	\begin{tabular}{|l|l|}
		\hline
		\textbf{Notation} & \textbf{Definition}  \\ \hline 
		$J$ & Number of job classes \\ \hline
		$M$ & Number of servers \\ \hline
		$\lambda_i$ & Arrival rate of class-$i$ jobs \\ \hline 
		$\lambda = \sum_i \lambda_i$ & Total system arrival rate \\ \hline
		$\mu_j$ & Service rate at server $j$ \\ \hline
		$\mu = \sum_j \mu_j$ & Total system service rate \\ \hline
		$S_i$ & The set of servers that can serve class-$i$ jobs \\ \hline
		$S(A)$ & The set of servers that can serve any job class in subset $A$ \\ \hline
		$C_j$ & The set of job classes that can be served by server $j$ \\ \hline
		$C(B)$ & The set of job classes that can be served by any server in subset $B$ \\ \hline \hline 
		$\vec{c}_n$ & Classes of all jobs in the queue, in order of arrival \\ \hline
		$\vec{b}_{l}$ & Busy servers, in the order in which they became busy \\ \hline
		$\vec{s}_k$ & Idle servers, in the order in which they became idle \\ \hline
		$\vec{d}_{l}$ & Classes of jobs in service, in the order in which they entered service \\ \hline
		$\vec{n}_{l}$ & Number of jobs not in service in between consecutive jobs in service \\ \hline
		$\vec{x}_J$ & Number of jobs of each class in the system \\ \hline
		$\vec{z}_m$ & Interleaving of jobs in the queue and busy servers, in order of job arrival times \\ \hline 
	\end{tabular}
	\caption{Summary of notation. Top section: system notation. Bottom section: notation used in state descriptors.}
	\label{tab:notation}
\end{table}

\section{Detailed States and Product Forms}
\label{sec:detailed_states}

We first consider the most complete descriptions of the state: the detailed state descriptor tracks the classes of all jobs in the order of their arrival, denoted by $\vec{c}_n$.
In the noncollaborative case, the two assigment rules that we consider (ALIS and RAIS) also require us to track some information about the servers.
Under ALIS (Section~\ref{sec:alis}), the state descriptor includes the vector $\vec{s}_k$, which tracks all idle servers in the order in which they became idle.
Under RAIS (Section~\ref{sec:rais}), the state descriptor is $\vec{z}_m$, which is an interleaving of $\vec{c}_n$ and $\vec{b}_{\ell}$, where $\vec{b}_l$ tracks all busy servers ordered by the arrival times of the jobs they are serving.

For both the collaborative and noncollaborative (ALIS and RAIS) models, we show that the stationary distribution for the above state descriptor exhibits a product form.
We begin with the collaborative case, which is a special case of what are known as \textquotedblleft Order Independent\textquotedblright (OI) queues, so named because the total service rate given to all jobs in the queue depends only on their classes, not on their order. 

At the end of the section, we discuss related models that also
have product-form stationary distributions.

\subsection{The Collaborative Model and Order Independent Queues\label{OI}}

The system state is $\vec{c}_{n}=(c_{1},\ldots ,c_{n})$, where $c_{i}$ is the class of the $i$'th job in the system in order of arrival, including both jobs that are in the queue and jobs that are in service (possibly at more than one server). 
The subscript $n$ can take on the values $0,1,...$; we will generally leave this implicit. 
Let ${\cal C}$ be the set of all such states. 
Abusing notation, let $S(\vec{c}_{n})=S(\{c_{1},\ldots ,c_{n}\})$ be the set of servers that can serve at least one of the jobs in the queue, and let 
\begin{equation}
\mu (\vec{c}_{n}):=\mu (\{c_{1},\ldots ,c_{n}\})=\sum_{m\in S(\vec{c}%
	_{n})}\mu _{m}  \label{collab mu}
\end{equation}%
be the total rate of service to jobs in the queue. Also, define $\Delta _{j}(%
\vec{c}_{n})$ as the (marginal) rate of service given to the $j$'th job in
the queue, so $\sum_{j=1}^{n}\Delta _{j}(\vec{c}_{n})=\mu (\vec{c}_{n})$,
and 
\[
\Delta _{j}(\vec{c}_{n})=\sum_{m\in S(\vec{c}_{j})\backslash S(\vec{c}%
	_{j-1})}\mu _{m}=\sum_{m\in S(\vec{c}_{j})}\mu _{m}-\sum_{m\in S(\vec{c}%
	_{j-1})}\mu _{m}=\mu (\vec{c}_{j})-\mu (\vec{c}_{j-1}). 
\]%

\begin{figure}
	\centering
	\includegraphics[scale=0.5]{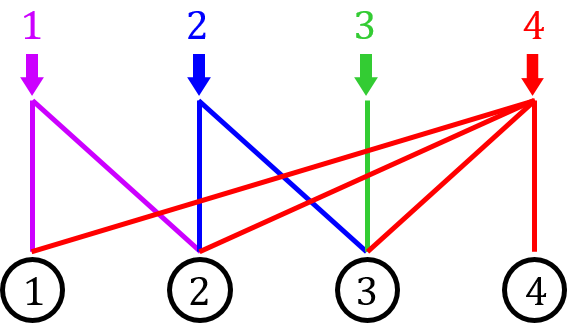} \hspace{1in}
	\includegraphics[scale=0.5]{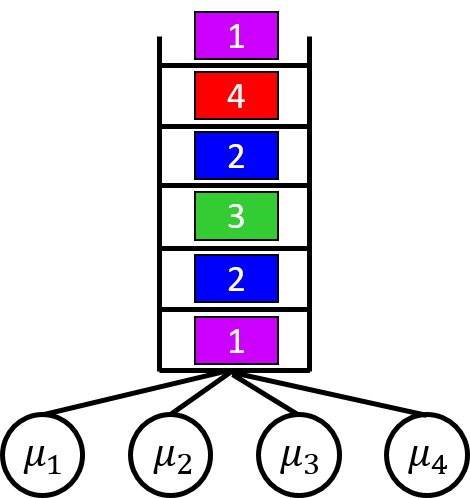}
	\caption{Left: a bipartite graph structure showing job/server compatibility constraints. Right: a system state in the collaborative model. In this example, the state is (1,2,3,2,4,1).}
	\label{fig:collab_ex}
\end{figure}

Figure~\ref{fig:collab_ex} shows an example of a possible state in the collaborative model.
In this example, the state is $\vec{c}_{n} = (1,2,3,2,4,1)$.
The class-1 job at the head of the queue is in service at both server 1 and server 2 ($\Delta_1(\vec{c}_{n}) = \mu_1 + \mu_2$), the class-2 job immediately behind it is in service at server 3 ($\Delta_2(\vec{c}_{n}) = \mu_3$), and the class-4 job is in service at server 4 ($\Delta_3(\vec{c}_{n}) = \Delta_4(\vec{c}_{n}) = 0$ and $\Delta_5(\vec{c}_{n}) = \mu_4$).
The total rate of service given to all jobs is $\mu(\vec{c}_n) = \mu_1 + \mu_2 + \mu_3 + \mu_4$.

Note that, for the collaborative model, the total service rate $\mu (\vec{c}_{n})$ is independent of the order of the jobs in the queue, and the service rate allocated to the $j$'th job doesn't depend on the jobs (if any) after job $j$ in the queue. 
That is, our collaborative model is a special case of an \emph{Order Independent} queue, defined as follows.

\begin{definition}\label{def:OI} A queue is said to be \emph{Order Independent} (OI) if it satisfies the following properties for all $\vec{c}_{n}$:
	\begin{enumerate}[label=(\roman*)]
		\item $\Delta _{j}(\vec{c}_{n})=$ $\Delta _{j}(\vec{c}_{j})$ for $j\leq n$,
		\item $\mu (\vec{c}_{n})$ is the same for any permutation of $c_{1},\ldots
		,c_{n}$,
		\item $\mu (c)>0$ for any class $c$.
	\end{enumerate}
\end{definition}

Properties (i)-(iii) are essentially the same as those defined by Krzesinski~\cite{K}, though Krzesinski's definition generalizes property (i) to also allow for an extra multiplicative service rate factor based on the number in queue. 
Our collaborative model can be generalized in this way to have {\em speed scaling}, i.e., a total service capacity of $\gamma (n)$ when there are $n$ jobs in the system.
Under this generalization, $\mu_{m}$ would be interpreted as the proportion of the total capacity used by server $m$, for $m\in S(\vec{c}_{n})$. 
The addition of a speed-scaling factor is straightforward, but complicates the notation, so we do not include it here. Similarly,
it is straightforward to include an arrival scaling (or rejection) factor, so that
arrivals of class $c$ occur according to a Poisson process with rate $r(n)\lambda_c$ when the number in queue is $n$, but, again, we do not include it for ease of exposition.

Property (iii) ensures irreducibility of the Markov chain. Property (ii) guarantees that the total rate of transitions out of any state $\vec{c}_{n}$ depends only on the set of customers in the queue and not on their order. A consequence of (i) and (ii) is that $\Delta _{j}(\vec{c}_{j})$ does not depend on the order of the first $j-1$ jobs. 
As Krzesinski shows, Properties (i)-(iii) are all that are needed to show that the stationary distribution has a product form. 
The proof below is essentially the same as Krzesinski's~\cite{K}; the special case for the collaborative model was shown by Gardner et al.~\cite{GZDMHS}.

We first recall the definition of quasi-reversibility.

\begin{definition}
	A queue is called quasi-reversible if its state at time $t$ is independent of
	
	\begin{itemize}
		\item arrival times after time $t$
		
		\item departure times before time $t$.
	\end{itemize}
\end{definition}

An equivalent definition is that the stationary distribution for the queue satisfies partial balance, i.e., for any state and any class $c$, the steady-state rate out of the state due to a class-$c$ arrival equals the steady-state rate into the state due to a class-$c$ departure, and the rate out of the state due to a departure equals the rate in due to an arrival.
Theorem~\ref{thm:collaborative pi} shows that the OI properties are sufficient for partial balance for the product-form distribution, and therefore for quasi-reversibility of the system.

\begin{theorem}
	\label{thm:collaborative pi}(Berezner, Kriel, Krzesinski~\cite{BKK}, Krzesinski~\cite{K}) 
	For any OI queue, including the collaborative model, the system is quasi-reversible and the stationary distribution is given by 
	\begin{equation}
	\pi^C (\vec{c}_{n})=\pi^C (\emptyset)\prod\limits_{i=1}^{n}\frac{\lambda _{c_{i}}}{\mu (\vec{c}_{i})}=\frac{\lambda _{c_{n}}}{\mu (\vec{c}_{n})}\pi^C (\vec{c}_{n-1})
	\label{product form}
	\end{equation}%
	as long as $G:=\sum_{n,\vec{c}_{n}\in {\cal C}}\prod\limits_{i=1}^{n}\frac{\lambda
		_{c_{i}}}{\mu (\vec{c}_{i})}<\infty $. Then $\pi^C \left( \emptyset \right) =1/G$ is
	the probability the system is empty.
\end{theorem}

\begin{proof}
	We will show that the product form of equation (\ref{product form})
	satisfies partial balance. First note that equation (\ref{product form})
	immediately satisfies the condition that the rate out of any state $\vec{c}%
	_{n}$ due to a departure equals the rate into the state due to an arrival: $%
	\mu (\vec{c}_{n})\pi^C (\vec{c}_{n})=\lambda _{c_{n}}\pi^C (\vec{c}_{n-1})$. Now
	we show that under the product-form probabilities (\ref{product form}), the
	rate out of any state $\vec{c}_{n}$ due to a class-$c$ arrival equals the
	rate into the state due to a class-$c$ departure, $\forall c$:%
	\begin{align*}
	\pi^C (\vec{c}_{n})\lambda _{c} &=\sum_{j=0}^{n}\pi^C (c_{1},\ldots
	,c_{j},c,c_{j+1},\ldots ,c_{n})\Delta _{j+1}(c_{1},\ldots
	,c_{j},c,c_{j+1},\ldots ,c_{n}) \\
	&=\sum_{j=0}^{n}\pi^C (c_{1},\ldots ,c_{j},c,c_{j+1},\ldots ,c_{n})\Delta
	_{j+1}(\vec{c}_{j},c) \qquad \text{ (Property (i))}\\
	&=\frac{\lambda _{c_{n}}}{\mu (\vec{c}_{n},c)}\sum_{j=0}^{n-1}\pi^C
	(c_{1},\ldots ,c_{j},c,c_{j+1},\ldots ,c_{n-1})\Delta _{j+1}(\vec{c}_{j},c) \\
	&\qquad \qquad \qquad \qquad +\frac{\lambda _{c}}{\mu (\vec{c}_{n},c)}
	\pi^C(\vec{c}_{n})\Delta _{n+1}(\vec{c}_{n},c)  \qquad \text{((2) and Property (ii)).}
	\end{align*}%
	We will show this by induction on $n$. For $n=0$, $\pi^C (0)\lambda _{c}=\pi^C
	(c)\mu (c)$ is immediate, given property (iii). Assume partial balance holds
	for the product-form probabilities (\ref{product form}) for any $\vec{c}%
	_{n-1}$, i.e.,
	\[
	\pi^C (\vec{c}_{n-1})\lambda _{c}=\sum_{j=0}^{n-1}\pi^C (c_{1},\ldots
	,c_{j},c,c_{j+1},\ldots ,c_{n-1})\Delta _{j+1}(\vec{c}_{j},c). 
	\]%
	Then we need to show%
	\begin{align*}
	\pi^C (\vec{c}_{n})\lambda _{c} &=\frac{\lambda _{c_{n}}}{\mu (\vec{c}_{n},c)}\sum_{j=0}^{n-1}\pi^C
	(c_{1},\ldots ,c_{j},c,c_{j+1},\ldots ,c_{n-1})\Delta _{j+1}(\vec{c}_{j},c)+%
	\frac{\lambda _{c}}{\mu (\vec{c}_{n},c)}\pi^C (\vec{c}_{n})\Delta _{j+1}(\vec{c%
	}_{n},c)\text{ }
	\end{align*}
	From the induction hypothesis, and the definition of $\Delta _{n+1}(\vec{c%
	}_{n},c)$, the right-hand-side is:
	\begin{align*}
	&\frac{\lambda _{c_{n}}}{\mu (\vec{c}_{n},c)}\pi^C (\vec{c}_{n-1})\lambda
	_{c}+ \frac{\lambda _{c}}{\mu (\vec{c}_{n},c)}\pi^C (\vec{c}_{n})[\mu (\vec{c}%
	_{n},c)-\mu (\vec{c}_{n})] \\
	&=\frac{\lambda _{c_{n}}}{\mu (\vec{c}_{n},c)}\pi^C (\vec{c}_{n-1})\lambda
	_{c}+\lambda _{c}\pi^C (\vec{c}_{n})-\frac{\lambda _{c}}{\mu (\vec{c}_{n},c)}%
	\frac{\lambda _{c_{n}}}{\mu (\vec{c}_{n})}\pi^C (\vec{c}_{n-1})\mu (\vec{c}%
	_{n}) \\
	&=\pi^C (\vec{c}_{n})\lambda _{c}.
	\end{align*}
\end{proof}

In the collaborative example in Figure~\ref{fig:collab_ex},
recalling that $\mu$ is the total service rate, the stationary probability of the depicted state is
$$\pi^C(\vec{c}_n) = \pi^C(\emptyset) \left( \frac{\lambda_1}{\mu_2 + \mu_3}  \right) \left( \frac{\lambda_2}{\mu_1 + \mu_2 + \mu_3}
\right) \left( \frac{\lambda_3}{\mu_1 + \mu_2 + \mu_3} \right) \left( \frac{\lambda_2}{\mu_1 + \mu_2 + \mu_3} \right) \left( \frac{\lambda_4}{\mu} \right) \left( \frac{\lambda_1}{\mu}  \right).$$

We reiterate that the skill-based collaborative queue is a special case of an OI\ queue. 
Other queues that are OI are the (noncollaborative) M/M/K queue with heterogeneous
servers, the M/M/$\infty $ queue, the M/M/1 queue under processor sharing, and the Multiserver Station with Concurrent Classes of Customers (MSCCC) queue~\cite{K}. 
The MSCCC queue is a multi-class M/M/K/FCFS queue with the restriction that at most $B_{c}$ customers of class $c$ can be in service (noncollaboratively) at the same time. 
The M/M/1/LCFS queue is{\em \ not} an OI queue even though it is a symmetric queue in the sense of Kelly, and is therefore quasi-reversible.

The following corollaries, generalizing the OI queue, follow immediately from Theorem~\ref{thm:collaborative pi}.

\begin{corollary}
	The departure process from an OI queue is a Poisson process; thus a network of OI queues will have a product-form stationary distribution.
\end{corollary}

Consider an order independent queue with abandonments, where a job of class $i$ abandons the system after an exponential time with rate $\gamma_i$.
This model also fits within the OI framework (i.e., properties (i)-(iii) are satisfied), so again has a product-form stationary distribution.

\begin{corollary}
	\label{abandon}
	In an order independent queue with abandonments,
	\begin{equation}
	\pi^C_A (\vec{c}_{n})=\pi^C_A (\emptyset)\prod\limits_{i=1}^{n}\frac{\lambda _{c_{i}}}{\mu (\vec{c}_{i})}=\frac{\lambda _{c_{n}}}{\mu (\vec{c}_{n})}\pi^C_A (\vec{c}_{n-1}),
	\end{equation}
	where
	\[
	\mu (\vec{c}_{j})=\sum_{i=1}^{j}\gamma _{c_{i}}+\sum_{m\in S(\vec{c}%
		_{j})}\mu _{m}
	\]%
	and
	\[
	\Delta _{j}(\vec{c}_{n})
	=\mu (\vec{c}_{j})-\mu (\vec{c}%
	_{j-1})=
	\gamma _{c_{j}}+\sum_{m\in S(\vec{c}_{j})\backslash
		S(\vec{c}_{j-1})}\mu _{m}\text{.} 
	\]%
\end{corollary}

As Berezner and Krzesinski \cite{BK} show, the product form result for OI queues also extends easily to OI loss models, where, following their terminology, we use the term loss in the general sense that arriving jobs may be rejected or lost,
depending on the current state. For the product-form
to continue to hold, the acceptance, or truncated, region
must satisfy the {\em truncation property}:
the job acceptance/rejection decision is also order independent and
rejection is more likely with more jobs. In particular, letting $%
{\cal C}_{T}$ comprise the states $(\vec{c}_{n},c)$ in which jobs of class $c$
are accepted when the state just before their arrival is $\vec{c}_{n}$, we have the following.

\begin{definition}\label{def:truncation}
	A set of states ${\cal C}_T$ satisfies the \emph{truncation property} if:
	\begin{description}
		\item[(i)] $\vec{c}_{n}$ $\in {\cal C}_{T}\Rightarrow {\cal P}(\vec{c}_{n}) \subseteq {\cal C}_{T}$, 
		where ${\cal P}( \vec{c}_{n} )$ denotes the set of permutations of $\vec{c}_{n}$, and
		
		\item[(ii)] $\vec{c}_{n}$ $\in {\cal C}_{T}\Rightarrow \vec{c}_{n-1}$ $\in {\cal C}_{T}$. \ 
		That is, using part (i) of the truncation property, if a job would be accepted with a given
		set of jobs in the queue, it will still be accepted if any job is removed
		from that set.
	\end{description}
\end{definition}

Letting $\vec{x}=(x_{1},...,x_{J})$ be the per-class aggregated state for $\vec{c}%
_{n}$ (which is sufficient for the acceptance/rejection decision because of
its OI property), the truncation property means the acceptance region for $x$
is coordinately convex. That is, the rejection decision is a threshold
decision, such that arrivals of type $c$ are rejected if $%
x_{c}>t(x_{1},...,x_{c-1},x_{c+1},...x_{J})$ for some function $t$. Simple
examples include having an upper bound on the total number of jobs, or having
upper bounds on the number in each job class.

The product form, now for $\vec{c}_{n}$ $\in {\cal C}_{T}$, is exactly the same,
except for the normalizing constant. In other words, the stationary
probability of being in a state in ${\cal C}_{T}$ for the loss model is the same
as the conditional probability of being in that state in the model without
losses, given that the state is in ${\cal C}_{T}$.
Let $\vec{C}\in{\cal C}$ be the random variable representing the state of the original collaborative system, with no rejections, in steady state, i.e., $\vec{C} \sim \pi^C$. Let 
$\vec{C}_{T}\in{\cal C}_{T}$ and $\pi^C_{T}$ be similarly defined for the model with rejections.

\begin{corollary}
	\label{loss}%
	For an OI queue with job rejection, if the acceptance region ${\cal{C}}_{T}$
	satisfies the truncation property, then 
	\[
	P\{\vec{C}_{T}=\vec{c}_{n}\}=
	P\{\vec{C}=\vec{c}_{n}|\vec{C}\in {\cal C}_{T}\}=\pi^C _{T}(\vec{c}_{n})=\pi^C
	_{T}(\emptyset )\prod\limits_{i=1}^{n}\frac{\lambda _{c_{i}}}{\mu (\vec{c}%
		_{i})}=\frac{\lambda _{c_{n}}}{\mu (\vec{c}_{n})}\pi^C (\vec{c}_{n-1})\text{
		for }\vec{c}_{n}\in {\cal C}_{T}\text{,}
	\]%
	where $\pi^C _{T}(\emptyset )=\pi^C (\emptyset)/P\{\vec{C}\in {\cal C}_{T}\}$.
\end{corollary}

\begin{proof}
	To see that $\pi_{T}$ has the given product form, note
	that for states and transitions to states in ${\cal C}_{T}$ the same partial balance equations hold as for the
	original OI\ queue, and for transitions where some of the states are not in $%
	{\cal C}_{T}$, the partial balance equations are easily seen to reduce to $0=0$
	because of the truncation property.
	For example, if $\vec{c}_{n}\in {\cal C}_{T}$, but $(\vec{c}_{n},c)\notin 
	{\cal C}_{T}$, then the rate out of $\vec{c}_{n}$ due to a class-$c$ arrival
	is 0, and the rate into $\vec{c}_{n}$ due to a class-$c$ departure is also
	0, because $\pi^C _{T}(\vec{c}_{n},c)=0$ for all permutations of $(\vec{c}%
	_{n},c)$. Also, for 
	$\vec{c}_{n} \in {\cal C}_{T}$,
	
	\[
	P\{\vec{C}=\vec{c}_{n}|\vec{C}\in {\cal C}_{T}\}=\frac{\pi^C (\emptyset
		)\prod\limits_{i=1}^{n}\frac{\lambda _{c_{i}}}{\mu (\vec{c}_{i})}}{%
		\sum_{j}\sum_{\vec{c}_{j}\in {\cal C}_{T}}\pi^C ( \vec{c}_{j}) }%
	=G\prod\limits_{i=1}^{n}\frac{\lambda _{c_{i}}}{\mu (\vec{c}_{i})}
	\]%
	where $G= \pi^C (\emptyset )/P\{\vec{C}\in {\cal C}_{T}\}$ is a
	normalizing constant, and, because of the form of $\pi^C _{T}$, $G=\pi^C_{T}(\emptyset )$.
\end{proof}

The following special cases will be useful later. Let the subscript $-A$
represent the system where all job classes in $A$ are removed. Let the
subscript ${\vdash B}$ represent a reduced system in which all the servers in set $B
$ are removed, as well as all job classes that are compatible with those
servers, i.e., job class $i$ is removed if $S_{i}\cap B\neq \emptyset $.
Note that if the original system is stable, such subsystems will also be stable.

\begin{corollary}
	\label{reduced}%
	\begin{description}
		\item[(i)] For all $\vec{c}_{n}\in {\cal C}_{-A}$%
		\[
		P\{\vec{C}=\vec{c}_{n}|\vec{C}\in {\cal C}_{-A}\}=P\{\vec{C}_{-A}=\vec{c}%
		_{n}\}=\pi^C _{-A}(\vec{c}_{n}). 
		\]
		
		\item[(ii)] For all $\vec{c}_{n}\in {\cal C}_{\vdash B}$%
		\[
		P\{\vec{C}=\vec{c}_{n}|\vec{C}\in {\cal C}_{\vdash B}\}=P\{\vec{C}_{\vdash B}=\vec{c}%
		_{n}\}=\pi^C _{\vdash B}(\vec{c}_{n}). 
		\]
	\end{description}
\end{corollary}

We mention here a recent extension of the OI queue by Comte and Dorsman~\cite{CD}, the "pass and swap" queue. In their model there is an undirected graph linking the classes of the OI queue, such that an edge between two classes
indicates that they are ``swappable.'' The service process satisfies
the conditions of the OI queue, but now completing (or replaced) jobs replace later, swappable jobs in the queue. A job that completes or is replaced and that finds no later swappable job leaves the system. They show that the same product form steady-state distribution holds for the pass and swap queue.

\subsection{The Noncollaborative Service Model: Assign Longest Idle Server}
\label{sec:alis}

We now turn to the noncollaborative model, in which a job is only allowed to enter service on one server and services are completed nonpreemptively. 
For this model we must also specify which server is used when an arriving job finds multiple idle and compatible servers; in this section we assume that this is according to Assign Longest Idle (compatible) Server (ALIS), and we will use the superscript $ALIS$ for the stationary distribution.

For the noncollaborative ALIS model we define the state as $(\vec{c}_{n},\vec{s}_{k})$
where $c_{i}$ is the class of the $i$'th oldest job that is {\em not}
receiving service, and $s_{i}$ is the idle server that has been idle $i$'th
longest (out of $k$ that are idle). 
Note that unlike in the collaborative model, here the $\vec{c}_n$ vector includes \emph{only} those jobs that are in the queue waiting for service (we will call this the job queue) and \emph{not} jobs that are in service.

\begin{figure}
	\centering
	\includegraphics[scale=0.5]{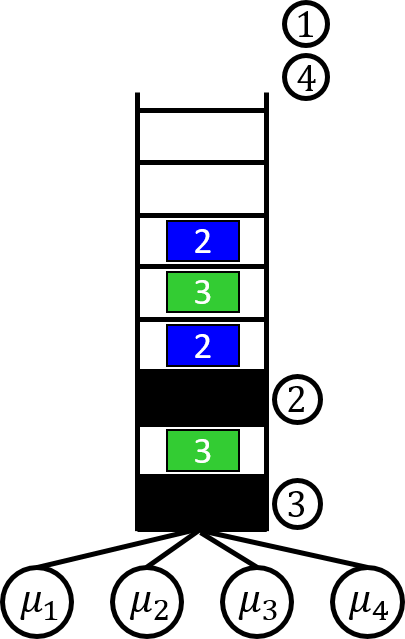}
	\caption{A system state in the noncollaborative model. In this example, the state is $(3,2,3,2;4,1)$.}
	\label{fig:noncollab_ex}
\end{figure}

Figure~\ref{fig:noncollab_ex} shows an example of a possible state in the noncollaborative ALIS model.
The state here is $(3,2,3,2;4,1)$. The class-3 job at the head of the queue is waiting to enter service on server 3, and the class-2 job immediately behind it is waiting to enter service on server 2 or server 3.
Servers 4 and 1 are idle, and server 4 became idle before server 1.
While our state does not explicitly record the positions of the busy servers within the job queue, we can infer that, for any job class $c$ that appears in the job queue, all servers in $S(c)$ are serving jobs that arrived earlier than that class-$c$ job.
For example, we can tell from the state of the job queue that server 2 is serving a job that arrived earlier than the first class-$2$ job in the queue.

We define the set of valid states, ${\cal X}^{ALIS}$, as those states $(\vec{%
	c}_{n},\vec{s}_{k})$ such that $s_{i}\notin S(\vec{c}_{n})$, $i=1,\ldots ,k$%
. That is, ${\cal X}^{ALIS}={\cal C}_{\vdash \vec{s}_{k}}^{C}\times {\cal S}$%
, where ${\cal S}$ is the set of all permutations of all subsets of $\{1,...,M%
\}$, and ${\cal C}_{\vdash \vec{s}_{k}}^{C}$ is the set of valid states for
the system queue (including jobs in service) for the reduced
collaborative model with the servers in $\vec{s}_{k}$ removed. Defining, as before, 
\[
\mu (\vec{c}_{n})=\sum_{m\in S(\vec{c}_{n})}\mu _{m}
\]%
we now have that $\mu (\vec{c}_{n})$ is the rate at which one of the first $n
$ jobs in the {\em job queue} leaves the queue (and enters service), and $%
\Delta _{j}(\vec{c}_{n})=\Delta _{j}(\vec{c}_{j})=\mu (\vec{c}_{j})-\mu (%
\vec{c}_{j-1})$ is the rate at which the $j$'th job in the job queue leaves
the queue (and enters service). The OI properties (i)-(iii) given in
Definition~\ref{def:OI} continue to hold for $\Delta _{j}(\vec{c}_{n})$ and $%
\mu (\vec{c}_{n})$. Indeed, given $\vec{s}_{k}$, the job queue is an OI loss
queue. In state $(\vec{c}_{n},\vec{s}_{k})$ an arrival of class $c$ will be
rejected from the job queue (and it will remove a server from the
idle-server queue) if $s_{i}\in S(c)$ for some $i=1,...,k$. The
state-dependent acceptance region for the job queue, ${\cal S}(\vec{s}_{k})$%
, satisfies the truncation property of the OI loss queue given in Definition~%
\ref{def:truncation}. 

We now consider the idle server queue.
Let $\lambda (\vec{s}_{j})$ be the rate of arrivals of jobs that are
compatible with one of the first $j$ (idle) servers, i.e., the rate of
departures from the idle server queue when it is in state $\vec{s}_{j}$. For 
$k\geq j$, let 
\[
\Delta _{j}^{\lambda }(\vec{s}_{k})=\lambda (\vec{s}_{j})-\lambda (\vec{s}%
_{j-1})=\sum_{i\in C(\vec{s}_{j})\backslash C(\vec{s}_{j-1})}\lambda
_{i}\geq 0
\]%
be the rate at which the $j$'th idle server will become busy (leave the idle
server queue). Note that we have the same OI properties (i)-(iii) for $\lambda (%
\vec{s}_{k})$ and $\Delta _{j}^{\lambda }(\vec{s}_{n})$ as we had for $\mu (%
\vec{c}_{n})$ and $\Delta _{j}(\vec{c}_{n})$:

(i) \ $\Delta _{j}^{\lambda }(\vec{s}_{k})=\Delta _{j}^{\lambda }(\vec{s}%
_{j})$ for $j\leq k$,

(ii) $\lambda (\vec{s}_{j})$ is the same for any permutation of $%
s_{1},\ldots ,s_{j}$ (order independence),

(iii) $\lambda (s)>0$ for any server $s$.

That is, given $\vec{c}_{n}$, the idle server queue is also an OI loss
queue, where we can think of servers of type $s$ arriving according to a Poisson
process at rate $\mu _{s}$, but if the server is already in the queue in
state $\vec{s}_{k}$, or if it will remain busy serving another job, i.e., if 
$s\in \vec{s}_{k}\cup S(\vec{c}_{n})$, then the arrival is rejected. Hence,
the acceptance region for the idle-server queue, given $\vec{c}_{n}$, also
satisfies the truncation property. 


The stationary distribution of the noncollaborative ALIS model has a
\textquotedblleft product of product forms\textquotedblright\ distribution,
with a product form component for the job queue and one for the idle server
queue.

\begin{theorem}
	\label{ALIS product form} (Adan et al. \cite{AKRW}) For the noncollaborative
	model, under FCFS for jobs and ALIS for servers, and given the stability
	condition, for $(\vec{c}_{n},\vec{s}_{k})\in {\cal X}$, 
	\begin{align}
	\pi ^{ALIS}(\vec{c}_{n},\vec{s}_{k})& =\pi ^{ALIS}(\emptyset ,\emptyset
	)\prod\limits_{i=1}^{n}\frac{\lambda _{c_{i}}}{\mu (\vec{c}_{i})}%
	\prod\limits_{j=1}^{k}\frac{\mu _{s_{j}}}{\lambda (\vec{s}_{j})} \\
	& =\frac{\lambda _{c_{n}}}{\mu (\vec{c}_{n})}\pi ^{ALIS}(\vec{c}_{n-1},\vec{s%
	}_{k})=\frac{\mu _{s_{k}}}{\lambda (\vec{s}_{k})}\pi ^{ALIS}(\vec{c}_{n},%
	\vec{s}_{k-1}),
	\end{align}%
	where $\pi ^{ALIS}(\emptyset ,\emptyset )$ is a normalizing constant equal
	to the probability that all servers are busy and that there are no jobs
	waiting in the queue.
\end{theorem}

\begin{proof}
	Note that if there is an arrival to or departure from the {\em job queue},
	the state of the idle servers do not change, and if there is an arrival to
	or departure from the set of idle servers, the state of the job queue does
	not change. Thus, the proof that the product form with $\pi ^{ALIS}(\vec{c}%
	_{n},\vec{s}_{k})=\frac{\lambda _{c_{n}}}{\mu (\vec{c}_{n})}\pi ^{ALIS}(\vec{%
		c}_{n-1},\vec{s}_{k})$ satisfies partial balance for job queue arrivals and
	departures, for fixed ${\vec s}_k$, is exactly the same as for our earlier proof for Theorem~\ref%
	{thm:collaborative pi}, using the truncation property
	of our acceptance region for job queue arrivals.
	For the idle server queue, fixing $\vec{c}_n$, the proof that the
	rate out of $(\vec{c}_{n},\vec{s}_{k})$ due to departures of idle servers
	equals the rate in due to arrivals of idle servers is immediate from $\pi
	^{ALIS}(\vec{c}_{n},\vec{s}_{k})=\frac{\mu _{s_{k}}}{\lambda (\vec{s}_{k})}%
	\pi ^{ALIS}(\vec{c}_{n},\vec{s}_{k-1})$. Finally, the proof that the rate of leaving
	state $\pi ^{ALIS}(\vec{c}_{n},\vec{s}_{k})$ due to server $s$ becoming idle
	(arriving to the idle server queue) equals the rate of entering the state
	due to server $s$ becoming busy (leaving the queue), for $s\notin \vec{s}_{k}
	$, is also very similar to our earlier induction proof. Note that to
	transition out of state $(\vec{c}_{n},\vec{s}_{k})$ due to the arrival of
	server $s$ to the idle server queue, $s$ must both be busy and not have any
	compatible jobs in the job queue, i.e., $s\notin \vec{s}_{k}\cup S(\vec{c}_{n})$.
\end{proof}

In the example shown in Figure~\ref{fig:noncollab_ex}, the stationary
probability is 
\[
\pi^{ALIS}(\vec{c}_n,\vec{s}_k) = \pi^{ALIS}(\emptyset,\emptyset) \left( 
\frac{\lambda_3}{\mu_3} \right) \left( \frac{\lambda_2}{\mu_2 + \mu_3}
\right) \left( \frac{\lambda_3}{\mu_2 + \mu_3} \right) \left( \frac{\lambda_2%
}{\mu_2 + \mu_3} \right) \left( \frac{\mu_4}{\lambda_4} \right) \left( \frac{%
	\mu_1}{\lambda_1 + \lambda_4 } \right).
\]

\subsection{The Noncollaborative Service Model: Random Assignment to Idle Servers}
\label{sec:rais}

We next consider the noncollaborative model where, instead of using ALIS to choose an idle server among compatible servers for an arriving job, servers are chosen randomly among idle compatible servers with appropriate probabilities that depend only on the set of busy (or idle) servers~\cite{VAW}. 
The results given in~\cite{VAW} use the partially aggregated state space we discuss in Section~\ref{sec:partial_agg}, but, as we will show, a product-form result also holds for a detailed state descriptor similar to the one used for the collaborative model.
Unlike under ALIS, under RAIS the order of the idle servers no longer matters.
Instead, for this version of the model we keep track of the busy servers, $\vec{b}_{l}$, where there are $l$ busy servers, and where the servers are ordered by the arrival times of the jobs they are serving.
To obtain a product-form stationary distribution, we need our state to be even more detailed: we must track not only the order of the busy servers, but the positions of the busy servers within the job queue.
Our detailed state description for RAIS is thus $\vec{z}_m$, where $m$ denotes the number of jobs in the system (including both jobs in the queue and jobs in service), and $z_{i}$ is associated with the $i$'th job in the system in order of arrival.
This is similar to the state $\vec{c}_n$ state used for the collaborative model, with one key difference: the $\vec{z}_m$ state does not track the classes of jobs that are in service, instead it tracks
the servers that are serving them.
That is, we let $z_{i}=c$ if the $i$'th job in the system has not started service (it is in the job queue), and $z_{i}=b$ if the $i$'th job in the system is in service on server $b$. Note that $\vec{z}_m$ consists of an interleaving of the states of the job queue, $\vec{c}_n$, and of the busy server queue, $\vec{b}_{l}$. The possible states for $\vec{z}_{m}$, $\mathcal{X}^{RAIS}$, are such that $\vec{b}_{l}\subseteq \{1,...,M\}$, and for any position $i$, if $z_{i}=c$, all compatible servers are serving earlier arrivals, i.e., $S(c)\subseteq \vec{z}_{i-1}$, because of the FCFS service discipline.

In order to completely define the RAIS policy, we must specify the probability that an arriving job enters service at compatible idle server $s \notin \{b_1,\dots,b_l\}$.
When the set of ordered busy servers is $\vec{b}_{j}$, let $\lambda _{s}^{a}(\vec{b}_{j})$ represent the \emph{activation rate} of idle server $s\notin \{b_{1},...,b_{j}\}$ (the rate of going from state $\vec{b}_{j}$ to $(\vec{b}_{j},s)$ for any $\vec{c}_{n}$). 
We allow the activation rates to depend only on the set of busy servers, not on their order. Indeed, as Visschers et al. showed for their aggregated state description of this model \cite{VAW}, in order for the stationary distribution to have a product form, we need the following stronger condition, called the \emph{assignment condition}. 
Let $\Pi _{\lambda }(\vec{b}_{l})=\prod\limits_{j=1}^{l}\lambda _{b_{j}}^{a}(\vec{b}_{j-1})$.
The assignment condition requires that the probabilities for routing to compatible idle servers be chosen so that $\Pi_{\lambda }(\vec{b}_{l})$ depends only on the set of busy servers, not on their order (i.e., so that $\Pi_{\lambda }(\vec{b}_{l})$ is the same for any permutation of $b_{1},\ldots ,b_{l}$).
Visschers et al. show that it is always possible to choose assignment probability distributions so that the assignment condition holds~\cite{VAW}; the derivation involves solving a max flow problem for each subset of busy servers.

One way to interpret the assignment condition is to consider the loss system in which customers are not allowed to queue, 
so that the state is just $\vec{b}_{l}$, the set of busy servers. 
Then the assignment condition, along with the fact that $\mu (\vec{b}_{l})$ doesn't depend on the order of busy servers, 
reduces to Kolmogorov's criterion for reversibility of Markov chains, namely that the product of the transition probabilities
along any path from a state back to itself is the same if the states are traversed in the reverse order. 
For example, consider the path traversing the states $\emptyset ,(u),(u,v),(v),\emptyset $, where $u$ and $v$ are two servers.
Then the probability of traversing that path is $C\lambda
_{u}^{a}(\emptyset )\lambda _{v}^{a}(u)\mu _{u}\mu _{v}$ where $C$ is an
appropriate normalizing constant, and the probability for the reverse path,
in which $v$ is activated first and finishes first, is $C\lambda
_{v}^{a}(\emptyset )\lambda _{u}^{a}(v)\mu _{u}\mu _{v}$. These are the
same, given the assignment condition: $\lambda _{u}^{a}(\emptyset )\lambda
_{v}^{a}(u)=\lambda _{v}^{a}(\emptyset )\lambda _{u}^{a}(v)$. Indeed, Adan, Hurkens, and Weiss showed that the loss model
(under the assignment condition) is reversible, and has a product-form stationary distribution~\cite{AHW}.

Let $\mu (\vec{z}_{i})=\sum_{j=1}^{i}I(z_{j})\mu _{z_{j}}$, where $I(z_{j})$ is an indicator that is 1 if $z_{j}$ corresponds to a busy server and 0 if it corresponds to a job in the job queue. 
Note that $\mu (\vec{z}_{m})$ satisifies the conditions for order independence, with $\Delta _{j}(\vec{z}_{m})=\Delta_{j}(\vec{z}_{j})=\mu (\vec{z}_{j})-\mu (\vec{z}_{j-1})$. 
Let $\lambda _{i}^{z}(\vec{z}_{i})=\lambda _{i}^{z}(\vec{z}_{m})=\lambda _{c}
$ if $z_{i}=c$ for some job class $c$, and, if $z_{i}=b$ for some busy
server $b$, let $\lambda _{i}^{z}(\vec{z}_{i})=\lambda _{i}^{z}(\vec{z}%
_{m})=\lambda _{b}^{a}(\vec{b}_{k(i)})$, where $k(i)=\sum_{j=1}^{i-1}I(z_{j})
$ is the number of busy servers in the first $i-1$ positions (the number of
servers serving the first $i-1$ arrivals). 
Note that $\lambda _{m}^{z}(\vec{z}_{m})$ is the same for any permutation of $\vec{z}_{m-1}$ regardless of whether $z_{m}$ is a waiting job or a busy server, and, from the assignment condition, for any $\vec{z}%
_{m}\in {\cal X}^{RAIS},$ $\lambda_{m-1} ^{z}(\vec{z}_{m-1})\lambda_m ^{z}(\vec{z}%
_{m})=\lambda_{m-1} ^{z}(\vec{z}_{m-2},z_{m})\lambda_m ^{z}(\vec{z}%
_{m-2},z_{m},z_{m-1})=\lambda_{m-1} ^{z}(\vec{z}_{m-2},z)\lambda_m ^{z}(\vec{z}_{m})$.

\begin{theorem}
	\label{thm:noncollab_rais_detailed} For the noncollaborative model, under
	FCFS for jobs and random assignment to idle servers, and under the
	assignment condition and the stability condition, for $\vec{z}_{m}\in {\cal X}$, 
	\[
	\pi ^{RAIS}(\vec{z}_{m})=\pi ^{RAIS}(\emptyset )\prod\limits_{i=1}^{m}\frac{%
		\lambda _{i}^{z}(\vec{z}_{i})}{\mu (\vec{z}_{i})}=\frac{\lambda _{m}^{z}(%
		\vec{z}_{m})}{\mu (\vec{z}_{m})}\pi ^{RAIS}(\vec{z}_{m-1}), 
	\]%
	where $\pi ^{RAIS}(\emptyset )$ is a normalizing constant that represents
	the probability that the system is empty (i.e., that there are no busy
	servers and no jobs in the queue).
\end{theorem}

Before proving Theorem~\ref{thm:noncollab_rais_detailed}, we give a brief example of the system state and stationary probability under RAIS. 
Consider the example in Figure 4. The state under RAIS is $\vec{z}_m = (3_b, 3_c, 2_b, 2_c, 3_c, 2_c)$, where we use a subscript of $b$ or $c$ to indicate whether the entry in $\vec{z}_m$ corresponds to a job that is waiting for service in the job queue (c) or to a busy server (b). Note that for jobs in the job queue the notation $i_c$ indicates that the job in this position is class-$i$, whereas for busy servers the notation $j_b$ indicates that server $j$ is in this position (that is, we do not track the classes of jobs in service). The state $\vec{z}_m$ is an interleaving of $\vec{c}_n = (3,2,3,2)$ and $\vec{b}_l = (3,2)$. The stationary probability for this state is
\[
\pi^{RAIS} (\vec{z}_m ) =
\pi^{RAIS}(\emptyset)\left(\frac{\lambda^a_{3}({\emptyset})}{\mu_3}\right)
\left ( \frac{\lambda_3}{\mu_3} \right )
\left ( \frac{\lambda^a_2(3)}{\mu_3+\mu_2} \right )
\left ( \frac{\lambda_2}{\mu_3+\mu_2} \right )
\left ( \frac{\lambda_3}{\mu_3+\mu_2}\right )
\left ( \frac{\lambda_2}{\mu_3+\mu_2}\right ).
\]

We are now ready to prove Theorem~\ref{thm:noncollab_rais_detailed}.

\begin{proof}
	Fix $\vec{z}_{m}\in {\cal X}^{RAIS}$, with corresponding $\vec{b}_{l}$. We will show that partial balance holds in three steps:
	\begin{enumerate}
		\item The rate out of state $\vec{z}_m$ due to a service completion equals the rate into state $\vec{z}_m$ due to an arrival.
		\item The rate out of state $\vec{z}_m$ due to server $b$ becoming busy equals the rate into state $\vec{z}_m$ due to server $b$ becoming idle.
		\item The rate out of state $\vec{z}_m$ due to a class-$c$ job arrival to the job queue equals the rate into state $\vec{z}_m$ due to a class-$c$ departure from the job queue.
	\end{enumerate}
	
	1. First suppose $z_{m}=b_{l}$, so $\vec{z}_{m}=(\vec{z}_{m-1},b_{l})$. Then our
	product form immediately satisfies $\mu (\vec{z}_{m})\pi ^{RAIS}(\vec{z}%
	_{m})=\lambda _{b_{l}}^{a}(\vec{b}_{l-1})\pi ^{RAIS}(\vec{z}_{m-1})$, i.e.,
	the rate of transitions out of state $\vec{z}_{m}$ due to a service
	completion equals the rate into $\vec{z}_{m}$ due to a new server arrival,
	i.e., of server $b_{l}$ going from idle to busy and serving the most
	recently arriving job. If $z_{m}\neq b_{l}$ then it is not possible to enter
	state $\vec{z}_{m}$ with an idle server becoming busy. Now suppose $%
	z_{m}=c_{n}$. In this case we have, for our product form, $\mu (\vec{z}%
	_{m})\pi ^{RAIS}(\vec{z}_{m})=\lambda _{c_{n}}\pi ^{RAIS}(\vec{z}_{m-1})$,
	i.e., the rate of transitions out of state $\vec{z}_{m}$ due to a service
	completion equals the rate into $\vec{z}_{m}$ due to a new arrival to the
	job queue. Note that $c_{n}$ is such that $S(c_{n})\subseteq \vec{b}_{l}$.
	
	2. We now show that under the product-form probabilities above, the rate out of
	state $\vec{z}_{m}$ due to the (external) arrival of any server $b\notin 
	\vec{b}_{l}$ to the busy server queue equals the rate into the state due to server $b$'s departure
	from the busy server queue. Note that because $\vec{z}_{m}\in {\cal X}^{RAIS}$ and $%
	b\notin \vec{b}_{l}$, none of the jobs in the job queue are compatible with
	server $b$, so a job completion at server $b$ in state $(z_{1},\ldots
	,z_{j},b,z_{j+1},\ldots ,z_{m})$ will result in server $b$ leaving the busy
	server queue. Using the OI\ properties of $\mu (\vec{z}%
	_{m})$, and that $\lambda_m ^{z}(\vec{z}_{m})$ is the same for any permutation
	of $\vec{z}_{m-1}$, we need to show that the given product form satisfies%
	\begin{align}
	\pi ^{RAIS}(\vec{z}_{m})
	\lambda_{m+1}^{z}(\vec{z}_{m},b)
	&= \sum_{j=0}^{m}\pi
	^{RAIS}(z_{1},\ldots ,z_{j},b,z_{j+1},\ldots ,z_{m})\Delta _{j+1}(\vec{z}%
	_{j},b) \nonumber \\
	&= \frac{\lambda_{m+1}^z (\vec{z}_{m-1},b,z_{m})}{\mu (\vec{z}_{m},b)}%
	\sum_{j=0}^{m-1}\pi ^{RAIS}(z_{1},\ldots ,z_{j},b,z_{j+1},\ldots
	,z_{m-1})\Delta _{j+1}(\vec{z}_{j},b) \nonumber \\
	&\qquad \qquad +
	\frac{\lambda_{m+1} ^{z}(\vec{z}_{m},b)}
	{\mu(\vec{z}_{m},b)}\pi ^{RAIS}(\vec{z}_{m})\Delta _{m}(\vec{z}_{m},b). \label{eq:RAIS1}
	\end{align}
	We use induction on $m$; the
	induction hypothesis is that 
	\[
	\pi ^{RAIS}(\vec{z}_{m-1})
	\lambda_m ^{z}(\vec{z}_{m-1},b)
	=\sum_{j=0}^{m-1}\pi
	^{RAIS}(z_{1},\ldots ,z_{j},b,z_{j+1},\ldots ,z_{m-1})\Delta _{j+1}(\vec{z}%
	_{j},b). 
	\]%
	Then the RHS of (\ref{eq:RAIS1}) is%
	
	\begin{align*}
	&\frac{\lambda_{m+1} ^{z} (\vec{z}_{m-1},b,z_{m})
	}{\mu (\vec{z}_{m},b)}\pi ^{RAIS}(\vec{z}%
	_{m-1})
	\lambda_m ^{z}(\vec{z}_{m-1},b)
	+\frac{\lambda_{m+1} ^{z}(\vec{z}_{m},b)}{%
		\mu (\vec{z}_{m},b)}\pi ^{RAIS}(\vec{z}_{m})[\mu (\vec{z}_{m},b)-\mu (\vec{z}%
	_{m})] \\
	&=\frac{
		\lambda_{m+1} ^{z} (\vec{z}_{m-1},b,z_{m})
	}{\mu (\vec{z}_{m},b)}\pi ^{RAIS}(\vec{z}%
	_{m-1})
	\lambda_m ^{z}(\vec{z}_{m-1},b)
	+
	\lambda_{m+1} ^{z}(\vec{z}_{m},b)
	\pi
	^{RAIS}(\vec{z}_{m}) \\
	&\qquad \qquad
	-\frac{\lambda_{m+1} ^{z}(\vec{z}_{m},b)}{\mu (\vec{z}_{m},b)}%
	\frac{\lambda_m ^{z}(\vec{z}_{m})}{\mu (\vec{z}_{m})}
	\pi ^{RAIS}(\vec{z}_{m-1})\mu (\vec{z}_{m}) \\
	&=\pi ^{RAIS}(\vec{z}_{m})
	\lambda_{m+1} ^{z}(\vec{z}_{m},b)
	\end{align*}
	where $\lambda_m ^{z}(\vec{z}_{m-1},b)\lambda_{m+1} (\vec{z}_{m-1},b,z_{m})=\lambda_m
	^{z}(\vec{z}_{m})\lambda_{m+1} ^{z}(\vec{z}_{m},b)$ from the assignment condition.
	
	3. Finally, we show that the rate out of $\vec{z}_{m}$ due to a class-$c$
	arrival to the job queue equals the rate in to state $\vec{z}_{m}$ due to a
	class-$c$ job queue departure, for each $c$ such that $S(c)\subseteq \vec{b}_{l}$.
	Fix $c$
	and $\vec{z}_m$ and call the class-$c$ job whose departure causes the system to enter state $\vec{z}_m$ the tagged job. Let $\vec{z}'_{m+1}$ denote the system state just before the tagged job leaves the job queue. The 
	transition from $\vec{z}'_{m+1}$ to $\vec{z}_m$ is triggered by a service completion at some server $b\in S(c)$. 
	In $\vec{z}'_{m+1}$ it must be the case that $b$, and all other servers
	in $S(c)$, are serving jobs that arrived earlier than the tagged job. At the service completion on server $b$, the job it is working on leaves, and server $b$ takes the position of the tagged job. Therefore,
	server $b$'s position in $\vec{z}_m$, after the service completion,
	must be after all the other servers in $S(c)$. Call this position $\kappa$. Before the transition, in state $\vec{z}'_{m+1}$, the tagged
	job must be in position $\kappa +1$, and server $b$ must be in position $j+1 \leq \kappa$.
	Thus, we need to show that 
	\begin{equation*}
	\pi ^{RAIS}(\vec{z}_{m})\lambda _{c}=\sum_{j=0}^{\kappa -1}\pi
	^{RAIS}(z_{1},\ldots ,z_{j},b,z_{j+1},..,z_{\kappa -1},c,z_{\kappa
		+1},\ldots ,z_{m})\Delta _{j+1}(\vec{z}_{j},b).
	\end{equation*}%
	First suppose $z_{m}=b$, i.e., $\kappa =m$. Then we want to show that 
	\begin{equation}
	\label{eq:RAIS2}
	\pi ^{RAIS}(\vec{z}_{m-1},b)\lambda _{c}=\sum_{j=0}^{m-1}\pi
	^{RAIS}(z_{1},\ldots ,z_{j},b,z_{j+1},..,z_{m-1},c)\Delta _{j+1}(\vec{z}%
	_{j},b).
	\end{equation}%
	Suppose, using induction on $m$, that 
	\[
	\pi ^{RAIS}(\vec{z}_{m-2},b)\lambda _{c}=\sum_{j=0}^{m-2}\pi
	^{RAIS}(z_{1},\ldots ,z_{j},b,z_{j+1},..,z_{m-2},c)\Delta _{j+1}(\vec{z}%
	_{j},b).
	\]%
	Note that for $j<m-1$,%
	\begin{align*}
	\pi ^{RAIS}(z_{1},\ldots ,z_{j},b,z_{j+1},..,z_{m-1},c) &= \frac{\lambda _{c}%
	}{\mu (\vec{z}_{m-1},b)}\frac{\lambda^z_{m-1}(\vec{z}_{m-1})}{\mu (\vec{z}_{m-1},b)}\pi
	^{RAIS}(z_{1},\ldots ,z_{j},b,z_{j+1},..,z_{m-2}) \\
	&= \pi ^{RAIS}(z_{1},\ldots ,z_{j},b,z_{j+1},..,z_{m-2},c,z_{m-1}) \\
	&= \frac{\lambda^z_{m-1}(\vec{z}_{m-1})}{\mu (\vec{z}_{m-1},b)}\pi ^{RAIS}(z_{1},\ldots
	,z_{j},b,z_{j+1},..,z_{m-2},c).
	\end{align*}%
	Thus, the RHS of (\ref{eq:RAIS2}) is%
	\begin{align*}
	&\frac{\lambda^z_{m-1}(\vec{z}_{m-1})}{\mu (\vec{z}_{m-1},b)}\sum_{j=0}^{m-2}\pi
	^{RAIS}(z_{1},\ldots ,z_{j},b,z_{j+1},..,z_{m-2},c)\Delta _{j+1}(\vec{z}%
	_{j},b)+\frac{\lambda _{c}}{\mu (\vec{z}_{m-1},b)}\pi ^{RAIS}(\vec{z}%
	_{m-1},b)\Delta _{m}(\vec{z}_{m-1},b) \\
	&=\frac{\lambda^z_{m-1}(\vec{z}_{m-1})}{\mu (\vec{z}_{m-1},b)}\pi ^{RAIS}(\vec{z}%
	_{m-2},b)\lambda _{c}+\frac{\lambda _{c}}{\mu (\vec{z}_{m-1},b)}\pi ^{RAIS}(%
	\vec{z}_{m-1},b)\left[ \mu (\vec{z}_{m-1},b)-\mu (\vec{z}_{m-1})\right]  \\
	&=\frac{\lambda^z_{m-1}(\vec{z}_{m-1})}{\mu (\vec{z}_{m-1},b)}\pi ^{RAIS}(\vec{z}%
	_{m-2},b)\lambda _{c}+\lambda _{c}\pi ^{RAIS}(\vec{z}_{m-1},b) - \frac{\lambda
		_{c}}{\mu (\vec{z}_{m-1},b)} \frac{\lambda^z_{m-1}(\vec{z}_{m-1})}{\mu (\vec{z}_{m-1})}%
	\pi^{RAIS}(\vec{z}_{m-2},b) \mu(\vec{z}_{m-1}) \\
	&= \lambda _{c}\pi ^{RAIS}(\vec{z}_{m-1},b).
	\end{align*}
	
	Now suppose $\kappa <m$, so $\vec{z}_{m}=(z_{1},...,z_{\kappa
		-1},b,z_{\kappa +1},...,z_{m})$. Then we want to show that 
	\begin{equation}
	\pi ^{RAIS}(z_{1},...,z_{\kappa -1},b,z_{\kappa +1},...,z_{m})\lambda
	_{c}=\sum_{j=0}^{\kappa -1}\pi ^{RAIS}(z_{1},\ldots
	,z_{j},b,z_{j+1},..,z_{\kappa -1},c,z_{\kappa +1},\ldots ,z_{m})\Delta
	_{j+1}(\vec{z}_{j},b), \nonumber
	\end{equation}%
	i.e.,%
	\begin{equation}
	\prod\limits_{i=\kappa +1}^{m}\frac{\lambda _{i}^{z}(\vec{z}_{i})}{\mu (\vec{%
			z}_{i})}\pi ^{RAIS}(\vec{z}_{\kappa -1},b)\lambda
	_{c}=\prod\limits_{i=\kappa +1}^{m}\frac{\lambda _{i}^{z}(\vec{z}_{i})}{\mu (%
		\vec{z}_{i})}\sum_{j=0}^{\kappa -1}\pi ^{RAIS}(z_{1},\ldots
	,z_{j},b,z_{j+1},..,z_{\kappa -1},c)\Delta _{j+1}(\vec{z}_{j},b). \nonumber
	\end{equation}%
	From our previous argument we have 
	\[
	\pi ^{RAIS}(\vec{z}_{\kappa -1},b)\lambda _{c}=\sum_{j=0}^{\kappa -1}\pi
	^{RAIS}(z_{1},\ldots ,z_{j},b,z_{j+1},..,z_{\kappa -1},c)\Delta _{j+1}(\vec{z%
	}_{j},b),
	\]%
	and the result follows.
\end{proof}

\subsection{Relationship between Collaborative and Noncollaborative Models}
\label{sec:c_nc_samplepaths}

The product form stationary probabilities for the collaborative model and the ALIS noncollaborative model both include the term $\prod_{i=1}^{n} \frac{\lambda_{c_i}}{\mu(\vec{c}_n)}$.
Given the similarities in the stationary distributions, it is natural to ask whether the two systems also are similar in their more detailed evolution.
Indeed, Adan et al.\ observed that when all servers are busy (i.e., the idle-server queue is empty, $\vec{s}_k = \emptyset$ under ALIS) the path-wise evolution of the state $\vec{c}_n$ (i.e., the jobs in queue) in the noncollaborative model is the same as the evolution of $\vec{c}_{n}$ (jobs in system) in the collaborative model~\cite{AKRW}.
We generalize this observation to relate the path-wise evolution of the two systems conditioned on the set of idle servers. Note that while
the set of idle servers is fixed, we need not worry about how jobs
are assigned to idle servers.

\begin{observation}
	\label{obs:samplepaths}
	Conditioned on the set of idle servers, $\vec{s}_{k}$, and while those servers remain idle, the path-wise evolution of $\vec{c}_{n}$ (jobs in queue) for the noncollaborative model (under either RAIS or ALIS) is the same as that of $\vec{c}_{n}$ (jobs in system) for the truncated collaborative model with the servers in $\vec{s}_{k}$ removed. 
\end{observation}

Observation~\ref{obs:samplepaths} tells us that, with coupled arrivals and service completions and the same initial $\vec{c}_n$, a service completion removes a job from the system for the collaborative model and removes the corresponding job from the job queue in the noncollaborative model. 
In the noncollaborative model, another job that does not appear in $\vec{c}_n$ will also leave the system (and will be replaced at the server by the job leaving the job queue). 

The path-wise correspondence between the collaborative and noncollaborative models will be useful when we move from the stationary distribution to performance metrics such as per-class response time distributions.
In Section~\ref{sec:nested}, we will see that these performance metrics often are more straightforward to derive in the collaborative model.
The path-wise relationship between the two models allows us to apply our results in the collaborative model to the noncollaborative model.

We note that the path-wise coupling still holds for general (coupled) arrival processes, not just Poisson processes. 

Our path-wise equivalence between the job queue in the noncollaborative model conditioned on the set of busy servers and the system queue for the collaborative model with the idle servers removed, is somewhat analogous to the observation of Borst et al. of the equivalence between the jobs in system in a processor sharing model with the jobs in queue for a nonpreemptive random-order-of-service model \cite{BBMN}.

\subsection{Token Models}

Two generalizations, combining aspects of the collaborative
and noncollaborative models, have recently been introduced
using the notion of ``tokens''~\cite{ABDV},~\cite{C}. In these models tokens generalize the notion of servers in the noncollaborative
model. There is a bipartite compatibility matching between
job classes and tokens, jobs must have tokens to enter service, and a token 
can be assigned to only one job at a time. Jobs of class $i$ arrive
according to a Poisson process at rate $\lambda _{i}$, and can be matched to
tokens in set $S_{i}$. 
Ayesta et al. allow jobs to wait for tokens and 
assume that when an arriving job sees multiple idle compatible tokens, it is assigned a token according to RAIS (or RAIT: Random Assignment to Idle Tokens)~\cite{ABDV}.
Comte assumes a loss model, in which jobs that arrive when
no compatible tokens are available are lost, and that idle tokens are assigned according to ALIS (or ALIT: Assign Longest Idle Token)~\cite{C}. We describe these models in more detail below.

\subsubsection{Token Model under RAIS}
In the model of Ayesta et al.~\cite{ABDV}, 
given the set of busy tokens $\vec{b}_{l}$, listed in the order of the arrival times of the jobs they are serving, and idle token $s$, the activation rate $\lambda _{s}^{a}(%
\vec{b}_{l})$ (i.e., the rate at which $s$ will be assigned to an arriving compatible
job) satisfies the same assignment condition as in the noncollaborative RAIS
model. The service process, given ordered busy tokens $\vec{b}_{l}$, is
generalized from the skill-based collaborative model to the OI\ queue. That
is, defining $\Delta _{j}(\vec{b}_{l})$ as the (marginal) rate of service
given to the job with the $j$'th busy token and $\mu (\vec{b}%
_{l})=\sum_{k=1}^{m}\Delta _{k}(\vec{b}_{l})$, the following OI conditions
are assumed, as in Definition~\ref{def:OI}:

(i) \ $\Delta _{k}(\vec{b}_{l})=$ $\Delta _{k}(\vec{b}_{k})$ for $j\leq l$,

(ii) $\mu (\vec{b}_{l})$ is the same for any permutation of $b_{1},\ldots
,b_{l}$ (order independence),

(iii) $\mu (b)>0$ for any busy token $b$.

Like Krzesinski \cite{K}, Ayesta et al. also allow the service rate to be
multiplied by a factor that is a function of the total number of tokens in
service. We continue to omit that factor for simplicity.

Let us define the state, as we did for the noncollaborative RAIS model, as $\vec{z}%
_{m}$ where $z_{i}$ is associated with the $i$'th arrival in the
system, $z_{i}=c$ if the arrival is of class $c$ and does not yet have a token, and $z_{i}=b$ if it has token $b$. We also define 
$\lambda_i(\vec{z}_i )$ as we did for the noncollaborative RAIS model.
Note that our proof of Theorem
\ref{thm:noncollab_rais_detailed} did not use the particular form of $\mu (\vec{b%
}_{l})$, only its OI\ properties. (In particular, we 
showed the result for general $\mu(\vec{z}_m)$ and $\Delta_j (\vec{z}_m)$).
Hence, Theorem \ref{thm:noncollab_rais_detailed} also holds for the token model.

As Ayesta et al. note~\cite{ABDV}, the noncollaborative model is recovered
when tokens correspond to the servers of the noncollaborative model, and the
original OI\ queue (including the collaborative model) is recovered when
each arriving job immediately obtains a token directly corresponding to its
class (so there is an infinite supply of tokens, and activation rates need
not be included).

\subsubsection{Token Loss Model under ALIS}

Comte introduced a related, multi-layered, token loss model that generalizes the
noncollaborative model operating under ALIS \cite{C}. The terminology and
notation used in Comte's model are a bit different from ours; Comte refers to job ``type'' where we use job ``class,'' and to token ``classes'' where we use ``tokens''. (We allow distinct 
tokens to have the same job class compatibilities and speeds.)
In Comte's model, and in contrast to
that of Ayesta et al., a job that arrives when there is no available
compatible token is lost, and a job that arrives to find multiple compatible
tokens takes the token that has been idle longest. Hence, for Comte, the
state is $(\vec{b}_{l},\vec{s}_{k})$ where $\vec{b}_{l}$ is the set of busy
tokens listed in the order of the arrival times of their corresponding jobs
and $\vec{s}_{k}$ is the set of idle tokens in the order in which they
became idle. Note that, because jobs cannot wait for tokens, there is no $\vec{c}_n$ component of the state, so $\vec{b}_{l}$ corresponds directly to $\vec{z}_m$ of the noncollaborative RAIS model. Also, tokens alternate between being busy and idle, and therefore $\vec{b}_{l}$ lists the busy tokens in the order in which they became busy, i.e., their order of arrival to the busy token queue.

Instead of assuming a generic OI service process $\mu (\vec{b}%
_{l})$ for serving tokens, as in Ayesta et al.'s model, Comte assumes a
collaborative service model. 
That is, there is another bipartite matching
layer between tokens and servers that defines the total service rate $\mu (%
\vec{b}_{l})$ when the ordered set of busy tokens is $\vec{b}_{l}$, and such
that $\mu (\vec{b}_{l})$ satisfies the OI conditions. Note that when the
idle token queue is in state $\vec{s}_{k}$, the rate at which tokens leave, $%
\lambda (\vec{s}_{k})$, is the rate at which jobs compatible with one of the
idle tokens arrive, and, as in the noncollaborative ALIS model, $\lambda (%
\vec{s}_{k})$ also satisfies the OI\ conditions \ref{def:OI}. Because there
is a finite set of tokens, we have a closed network of two OI queues, and
because OI\ queues are quasi-reversible, the closed token (CT) network also
has a product-form distribution. In particular, for $(\vec{b}_{l},\vec{s}%
_{k})\in {\cal X}^T$, where ${\cal X}^T$ is the set of states such that each
token appears exactly once, i.e., $l+k$ equals the total number of tokens
and $(\vec{b}_{l},\vec{s}_{k})$ is an arbitrary permutation of the set of
tokens, we have 
\[
\pi ^{CT}(\vec{b}_{l},\vec{s}_{k})=G^{CT}\prod\limits_{i=1}^{l}\frac{1}{\mu (%
	\vec{b}_{i})}\prod\limits_{i=1}^{k}\frac{1}{\lambda (\vec{s}_{i})}
\]%
where $G^{CT}$ is a normalizing constant.
This result would also hold assuming a general OI process for
``serving'' busy tokens rather than the collaborative service
model.

\subsection{Discrete-time OI Queues and Matching Models}

Adan et al.~\cite{AKRW} introduced a matching model, called the
directed bipartite matching (DBM) model, with a bipartite
matching between servers and jobs, and in which both servers and jobs arrive
according to Poisson processes, but only jobs can queue to wait for servers. 
Servers of type $j$ arrive according to an independent Poisson process with
rate $\mu _{j}$, and the other parameters of the model are the same as for
the collaborative model. The state is again $\vec{c}_{n}$. An arriving
server matches with the first compatible job in the queue, if any, and the
server, along with its job if there is one, immediately leaves. 
The DBM model captures important features of organ transplant waitlists, where patients
wait for organs, but unmatched organs are lost, and where compatibilities
are determined by biological factors such as blood types, as well as
the locations of the patients and organs.
As Adan et
al. show, the Markov chain for this model is sample-path
equivalent to that of the collaborative model; in particular, the departure
rate from the queue in state $\vec{c}_{n} $ is $\mu (\vec{c}_{n})$ as
defined earlier. Therefore the matching model has the same, product-form,
stationary distribution given in Theorem \ref{thm:collaborative pi}. The result
also holds for a more general, OI matching, i.e., when there are no server
types, but a job will be matched to a server at rate $\mu (\vec{c}_{n})$ when the state
is $\vec{c}_{n}$ and $\mu (\vec{c}_{n})$ satisfies the OI\ conditions
(i)-(iii). The state process for the matching model is also 
equivalent to the queue process of a variant of the
noncollaborative model in which we keep all the servers busy by assigning a
server that becomes idle and that does not find a waiting compatible job a
\textquotedblleft dummy job.\textquotedblright\ This might be
appropriate in a call center context in which servers that would
be otherwise idle handle outgoing calls or email.

If we ignore the timing between arrivals and departures in
the matching model described above, we have an
equivalent discrete-time model, in which at most one event (a job arrival or
a server arrival/job departure) can occur in any time slot. Now $\lambda
_{c} $ is the probability of a class-$c$ arrival and, when the state is $%
\vec{c}_{n}$, $\mu (\vec{c}_{n})$ represents the probability of a job
completion (or a job-server matching) in the next time slot, $\mu (\vec{c}%
_{n})\leq 1$. Also, we need not assume a set of servers with a
bipartite-matching graph, just that $\mu (\vec{c}_{n})$ satisfies the OI
conditions. Then the transitions of the Markov chain $\vec{c}_{n}$ for the
discrete-time queue will be sample-path equivalent to the transitions of the
embedded Markov chain for the continuous-time OI queue, and, again, the same
product form will hold for the steady-state distribution. The DBM special case, 
with server/job compatibilities, is considered by Weiss \cite{W}. 
Here a server of type $s$
arrives with probability $\mu _{s}$ and matches with the earliest compatible
job if there is one; the server immediately departs (along with any matching
job).

The DBM model discussed in~\cite{AKRW} does not allow unmatched servers to wait for jobs.
We now show that the DBM model can be extended to include a ``server queue'' in which  unmatched servers wait in FCFM (first-come-first-matched) order.
This yields something somewhat analogous to
the noncollaborative ALIS model. For stability, we must have an upper bound, 
$K$, on the server queue. Let us call this (new) model the DBM($K$) model. Then
the stability condition for the DBM($K$) model will be the same as for the
DBM =DBM($0$) model, i.e., $\lambda (A)\leq \mu (A)$ for all subsets of job
classes $A$, where $\mu (A)$ is the rate of arrivals of servers compatible
with job classes in $A$ in the continuous-time model, and is the probability
of such an arrival in the discrete-time version. The state is $(\vec{c}_{n},%
\vec{s}_{k})\in {\cal X}^{DBM(K)}$, where $\vec{c}_{n}$ is the set of waiting jobs in
arrival order, $\vec{s}_{k}$ is the set of waiting servers in arrival order,
and the set of valid states, ${\cal X}^{DBM(K)}$, comprises those states $(\vec{c}%
_{n},\vec{s}_{k})$ such that $k\leq K$ and $s_{i}\notin S(\vec{c}_{n})$, $%
i=1,\ldots ,k$. Again, both the job queue and the server queue are order
independent, so the steady-state distribution will have the same form as
that of the noncollaborative ALIS model,
though the latter includes a particular
loss model for the idle-server queue, so its set of valid states is
restricted to $\vec{s}_{k}$ such that each server appears in $\vec{s}_{k}$
at most once.

\begin{theorem}
	For the stable directed bipartite matching model with a finite buffer for
	servers $K$, DBM($K$), 
	\[
	\pi ^{DBM(K)}(\vec{c}_{n},\vec{s}_{k})
	=\pi ^{DBM(K)} (\emptyset,\emptyset)\prod_{i=1}^{n}\frac{\lambda _{c_{i}}}{\mu (\vec{c}_{i})}%
	\prod_{j=1}^{k}\frac{\mu _{s_{j}}}{\lambda (\vec{s}_{j})},\forall (\vec{c}_{n},\vec{s}_{k})\in {\cal X}^{DBM(K)}. 
	\]
\end{theorem}

By symmetry, a similar result holds if the server buffer is infinite, but
the job buffer is bounded by some $N$. \ Now the stability condition is $\mu
(B)\leq \lambda (B)$ for all subsets of server types $B$. 

Note that the continous-time DBM(K) model also models a make-to-stock
inventory system with a bipartite graph representing preferences of
customer classes for certain types of items. Customers of class $i$
are willing to purchase any of the items in $S_i$. Items of type $j$ are produced according to a Poisson process at rate $\mu_j$ as long as the total number of items is less than the overall base-stock level $K$. Queueing customers represent back orders. Also, from \ref{loss}, the result holds when we have different base-stock levels for different types of items.

Our results also extend to DBM models with abandonments and finite or infinite buffers. These models are appropriate for car sharing applications and other two-sided queues, where, for example, classes of jobs and types of servers correspond to location preferences. Suppose jobs (riders) of class $i$ arrive (request rides) according to a Poisson process with rate $\lambda_i$, and will wait for an exponential time at rate $\gamma_i$ before abandoning their request. Servers (drivers) of type $j$ arrive according to a Poisson process at rate $\mu_j$ and will wait an exponential time at rate $\nu_j$ for a rider before leaving the platform. We assume a bipartite matching graph as defined earlier. Because of the abandonments, stability will not be an issue, even for infinite buffers. We have the following. 

\begin{theorem}
	For the directed bipartite matching model with abandonments (DBMA) and finite or infinite buffers for jobs and servers,
	\[
	\pi ^{DBMA}(\vec{c}_{n},\vec{s}_{k})
	=\pi ^{DBMA} (\emptyset,\emptyset)\prod_{i=1}^{n}\frac{\lambda _{c_{i}}}{\mu (\vec{c}_{i})}
	\prod_{j=1}^{k}\frac{\mu _{s_{j}}}{\lambda (\vec{s}_{j})},\forall (\vec{c}_{n},\vec{s}_{k})\in {\cal X}^{DBMA}, 
	\]
	where
	\[
	\mu (\vec{c}_{j})=\sum_{i=1}^{j}\gamma _{c_{i}}+\sum_{m\in S(\vec{c}%
		_{j})}\mu _{m}\text{, \ }
	\lambda (\vec{s}_{j})=\sum_{i=1}^{j}\nu _{s_{i}}+\sum_{m\in C(\vec{s}%
		_{j})}\lambda _{m},
	\]
	and ${\cal X}^{DBMA}$ is the set of states $(\vec{%
		c}_{n},\vec{s}_{k})$ such that $s_{j}\notin S(\vec{c}_{n})$, $j=1,\ldots ,k$ and $c_{i}\notin S(\vec{s}_{k})$, $i=1,\ldots ,n$.
	
\end{theorem}

Moyal, Bu\v{s}i\'{c}, and Mairesse show reversibility and a product-form
stationary distribution for a FCFM matching model with a General (not
necessarily bipartite) Matching (GM) graph, and with sequential individual
(non-paired) arrivals, under a given stability condition \cite{MBM}. For
this model, instead of jobs and servers we have \textquotedblleft
agents\textquotedblright\ of $J$ different classes, with agent classes
corresponding to nodes in the compatibility graph; the set of agent classes
compatible with class $c$, $S(c)$, is its set of neighbors in the
compatibility graph. The set of valid states, ${\cal C}^{GM}$, are those states $%
\vec{c}_{n}$ such that $c_{i}\notin S(\vec{c}_{n})$, $i=1,\ldots ,n$. Among
the arrival processes Moyal et al.\ consider is i.i.d.\ arrivals where the
probability of a class-$c$ arrival is $\mu _{c}$.  Given the classes of
unmatched agents ordered by their arrival times, $\vec{c}_{n}$, let $\mu (%
\vec{c}_{n})$ be the probability the next arrival is compatible with one of
those agents. Again, $\mu (\vec{c}_{n})$ satisfies the OI\ conditions (now
in discrete time), so the stationary distribution for the GM model, assuming stability, is%
\[
\pi ^{GM} (\vec{c}_{n})=\pi ^{GM} (\emptyset )\prod_{i=1}^{n}\frac{\mu _{c_{i}}%
}{\mu (\vec{c}_{i})}=\frac{\mu _{c_{n}}}{\mu (\vec{c}_{n})}\pi (\vec{c}%
_{n-1})\text{ \ for }\vec{c}_{n}\in {\cal C}.
\]

Adan and Weiss consider the Paired Bipartite Matching (PBM) model in which
server-job pairs arrive sequentially and where the job is type $i$ and the
server is type $j$ independently and with respective probabilities $\lambda
_{i}/\lambda $ and $\mu _{i}/\mu $, and both unmatched jobs and unmatched
servers wait for matches \cite{AW}. An arriving job (server) is matched to
the first compatible waiting server (job) if there is one and they both
immediately leave, otherwise the job (server) waits for a match. Adan and
Weiss show that the associated Markov chain satisfies partial balance and
has a product-form stationary distribution, under the stability condition.
That is,

\[
\pi ^{PBM} (\vec{c}_{n},\vec{s}_{n})=\pi ^{PBM} (\emptyset )\prod_{i=1}^{n}\frac{%
	\lambda _{c_{i}}}{\mu (\vec{c}_{i})}\frac{\mu _{s_{i}}}{\lambda (\vec{s}_{i})%
}. 
\]%
Note that for the PBM model, there are always the same number of unmatched
jobs and servers. Adan et al. show that there exists a unique FCFM
(first-come first-matched) matching for the PBM model, and that the process
is reversible under an \textquotedblleft exchange
transformation\textquotedblright\ that interchanges matching servers and
customers \cite{ABMW}. 

\section{Nested Systems and Response Time Distributions}
\label{sec:nested}

In the previous section we developed product forms for the stationary
distributions of the detailed states for variants of OI\ queues, but these product forms do not readily yield other important performance measures, such as response time distributions. It turns out that we will get simple, elegant
results for response times in the collaborative model for a particular system structure called a nested system (see Figure~\ref{fig:nested} for an example). As noted in Observation~\ref{obs:samplepaths}, conditioned on the set of busy servers the noncollaborative
queue state has the same sample-path evolution as the collaborative system state for a system with only the busy servers available.
A consequence of this result is that our results for collaborative response times (Section~\ref{sec:nested_collab}) carry over to noncollaborative queueing times (Section~\ref{sec:nested_noncollab}). 

Formally, a nested system is one in which, for any two job classes $i\neq j$, the sets of servers with
which they are compatible, $S_i$ and $S_j$, are such that $S_{i}\subset
S_{j}$ or $S_{j}\subset S_{i}$ or $S_{i}\cap S_{j}=\emptyset $. This means
that nested systems can be recursively defined, starting with their most
flexible job class, as follows.

\begin{figure}
	\centering
	\includegraphics[scale=0.6]{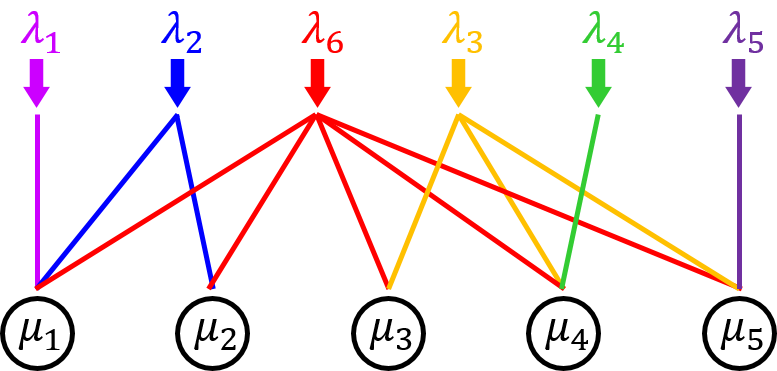}
	\caption{An example of a nested system.}
	\label{fig:nested}
\end{figure}

All nested systems have a most flexible job class, $i$, that is compatible
with all the servers in the system, and if we remove class $i$ from the
system it decomposes into two or more nonoverlapping nested subsystems, each
with its own fully flexible job class. These in turn can be decomposed by
removing the fully flexible class until we get down to systems consisting of
single job classes. Figure~\ref{fig:nested} shows an example of a nested system; if the fully flexible class 6 is removed, the system decomposes into one nested system consisting of servers 1 and 2 and job classes 1 and 2, and another nested system consisting of servers 3, 4, and 5 and job classes 3, 4, and 5.

We begin our response time derivations with the collaborative model, and first determine the response time of a class that is fully flexible, which, as we will see, has an exponential
distribution. The derivation for the fully flexible class does not require
the system to be nested, but later it will help us to develop general
response times in nested systems. We note that the results for nested systems were first derived by Gardner et al.~\cite{GHHR} using an alternative state descriptor specific to nested systems; here we provide a new derivation that follows directly from the detailed states used in Section~\ref{sec:detailed_states}.

\subsection{Collaborative Model}
\label{sec:nested_collab}

\subsubsection{Fully flexible class}

Let ${\cal C}=\{\vec{c}_{n},n=0,1,...\}$ be the set of all states $\vec{c}%
_{n}$ for the original collaborative model, and let $\vec{C}$ be a random
variable representing the state of the collaborative system in steady state,
i.e., $\vec{C}\sim \pi $. We use the subscript $-i$ to represent a reduced
system without class $i$, $i=1,...,J$. From
Corollary \ref{loss}, we have that for $\vec{c}_{n}\in {\cal C}_{-i}$, 
\[
P\{\vec{C}=\vec{c}_{n}|\vec{C}\in {\cal C}_{-i}\}=\pi^C_{-i}(\vec{c}_{n})=\pi^C
_{-i}(\emptyset )\prod\limits_{i=1}^{n}\frac{\lambda _{c_{i}}}{\mu (\vec{c%
	}_{i})}, 
\]%
where $\pi^C_{-i}(\emptyset )=$ $\pi^C(\emptyset )/P\{\vec{C}\in {\cal C}%
_{-i}\}$.

Suppose there is one class, call it class $J$, that is fully flexible in the
bipartite compatibility matching, i.e., $S_{J}=\{1,\ldots ,M\}$. We
condition on there being at least one class-$J$ job in the system, $\vec{C}%
\in {\cal C}\backslash {\cal C}_{-J}$, so we know all servers will be busy.
A possible state is $(\vec{c}_{n},J,\vec{a}_{m})$ where the first class-$J$
job is in position $n+1$, $\vec{c}_{n}\in {\cal C}_{-J}$ represents the
classes of jobs ahead of the first class-$J$ job in order of arrival, and $%
\vec{a}_{m}\in {\cal C}$ represents the classes of jobs after the first
class-$J$ job in order of arrival. Then, because the denominator for all the
terms corresponding to the first class-$J$ job and the jobs after it is the
total service rate $\mu $, we have 
\[
\pi ^{C}(\vec{c}_{n},J,\vec{a}_{m})=\pi ^{C}(\emptyset
)\prod\limits_{i=1}^{n}\frac{\lambda _{c_{i}}}{\mu (\vec{c}_{i})}\frac{%
	\lambda _{J}}{\mu }\prod\limits_{k=1}^{m}\frac{\lambda _{a_{k}}}{\mu }.
\]%
Let $\vec{C}_{before}$ be the conditional state of the jobs before the first
class-$J$ job, given there is such a job. Then, for $\vec{c}_{n}\in {\cal C}%
_{-J}$, 
\[
P\{\vec{C}_{before}=\vec{c}_{n}\}=\frac{\pi ^{C}(\emptyset
	)\prod\limits_{i=1}^{n}\frac{\lambda _{c_{i}}}{\mu (\vec{c}_{i})}\left( 
	\frac{\lambda _{J}}{\mu }\sum_{m,\vec{a}_{m}\in {\cal C}}\prod%
	\limits_{k=1}^{m}\frac{\lambda _{a_{k}}}{\mu }\right) }{\left( \frac{\lambda
		_{J}}{\mu }\sum_{m,\vec{a}_{m}\in {\cal C}}\prod\limits_{k=1}^{m}\frac{%
		\lambda _{a_{k}}}{\mu }\right) \sum_{j,\vec{c}_{j}\in {\cal C}_{-J}}\pi
	^{C}(\emptyset )\prod\limits_{i=1}^{n}\frac{\lambda _{c_{i}}}{\mu (\vec{c}%
		_{i})}}=\frac{\pi ^{C}(\emptyset )}{P\{\vec{C}\in {\cal C}_{-J}\}}%
\prod\limits_{i=1}^{n}\frac{\lambda _{c_{i}}}{\mu (\vec{c}_{i})}=\pi
_{-J}^{C}(\vec{c}_{n}).
\]%
That is, the first class-$J$ job \textquotedblleft sees\textquotedblright\
the steady-state distribution for the collaborative model with class $J$
removed. Similarly, letting $\vec{C}_{after}$ be the conditional state for
the jobs after the first class-$J$ job, given there is one, we have 
\[
P\{\vec{C}_{after}=\vec{a}_{m}\}=A\prod\limits_{k=1}^{m}\frac{\lambda
	_{a_{k}}}{\mu },
\]%
where the normalizing constant is 
\[
A=\left( \sum_{m,\vec{a}_{m}\in {\cal C}}\prod\limits_{k=1}^{m}\frac{\lambda
	_{a_{k}}}{\mu }\right) ^{-1}=\left( \sum_{m=0}^{\infty
}\prod\limits_{k=1}^{m}\sum_{i=1}^{J}\frac{\lambda _{i}}{\mu }\right)
^{-1}=\left( \sum_{m=0}^{\infty }\frac{\lambda }{\mu }^{m}\right)
^{-1}=(1-\rho )
\]%
with $\rho =\frac{\lambda }{\mu }$. Also, given there is at least once class-%
$J$ job, $\vec{C}_{before}$ and $\vec{C}_{after}$ are independent. Finally,
letting $N^{J}$ be the total number of class-$J$ jobs in the system in
steady state, we have  
\begin{align*}
	P\{N^{J} \geq 1\}&=\frac{\lambda _{J}}{\mu }\sum_{j,\vec{c}_{j}\in {\cal C}%
		_{-J}}\pi ^{C}(\emptyset )\prod\limits_{i=1}^{n}\frac{\lambda _{c_{i}}}{\mu (%
		\vec{c}_{i})}\sum_{m,\vec{a}_{m}\in {\cal C}}\prod\limits_{k=1}^{m}\frac{%
		\lambda _{a_{k}}}{\mu } \\
	&=\frac{\lambda _{J}}{\mu }P\{N^{J}=0\}\frac{1}{1-\rho }.
\end{align*}%
Solving for $P\{N^{J}\geq 1\} = 1 - P\{N^{J} = 0\}$, we obtain $P\{N^{J}=0\} = \rho _J$ where $\rho _{J}=\frac{\lambda _{J}}{\mu -(\lambda -\lambda
	_{J})}$.

Note that $P\{\vec{C}_{after}=\vec{a}_{m}\}$ is the same as the probability
of state $\vec{a}_{m}$ in a multiclass M/M/1 queue with service rate $\mu $.
Let $\hat{N}$ be total number of jobs after the first class-$J$ job in
steady state. From standard results for the M/M/1 queue, we have that $\hat{N%
}\sim geom(1-\rho )$, where $Y\sim geom(p)$ means $P\{Y=n\}=p(1-p)^{n}$, $%
n=0,1,...$. We can also obtain this result by summing the product form
result above: $P\{\hat{N}=n\}=\sum_{\vec{a}_{n}\in {\cal C}}P\{\vec{C}%
_{after}=\vec{a}_{n}\}$. Each of the $\hat{N}$ jobs is independently class $i
$ with probability $\lambda _{i}/\lambda $, so $\hat{N}^{J}$, the number of
class-$J$ jobs after the first class-$J$ job, is also geometrically
distributed, $\hat{N}^{J}\sim geom(1-\rho _{J})$. 
More generally, $\hat{N}^{i}$,
the number of class-$i$ jobs after the first class-$J$ job has a geometric
distribution, $\hat{N}^{i}\sim geom(1-\frac{\lambda _{i}}{\mu -\lambda
	+\lambda _{i}})$. This is a consequence of the following simple lemma
regarding Bernoulli splitting of geometric random variables, with $p=\rho
=\lambda /\mu $ and $q_{i}=\lambda _{i}/\mu $; we include the proof for
completeness.

\begin{lemma}
	\label{split}Let $Y\sim geom(1-p)$, i.e., $Y$ is the number of failures
	before the first success in i.i.d. Bernoulli trials with failure probability 
	$p$. Let $Y_{i}$ be the number of type-$i$ failures before the first success
	in i.i.d.\ Bernoulli trials with success probability $1-p$ and type-$i$
	failure probability $q_{i}$, with $\sum q_{i}=p$, so $Y_{i}|Y\sim
	Binomial(Y,q_{i}/p)$. Then $Y_{i}\sim geom(1-q_{i}/(q_{i}+1-p))$.
\end{lemma}

\begin{proof}
	When we are counting the number of type-$i$ failures before the first
	success, we can ignore the other types of failures. That is, we can just
	look at the trials that result in either type-$i$ failures or success.
	Conditioned on the trial being either a success or a type-$i$ failure, the
	probability that it is a type-$i$ failure is $q_{i}/(q_{i}+1-p)$.
\end{proof}

Because $N^{J}=I\{N^{J}>0\}(\hat{N}^{J}+1)$, and $\hat{N}^{J}\sim
geom(1-\rho _{J}$), and, as we showed above, $P\{N^{J}=0\}=\rho _{J}$, we
have the following.

\begin{corollary}
	$N^{J}\sim geom(1-\rho _{J})$.
\end{corollary}

Summarizing our observations so far, we have the following.

\begin{theorem}
	\label{fully flexible}For the collaborative model with a fully flexible job
	class $J$,
	
	\begin{description}
		\item[(i)] The steady-state distribution for the system conditioned on there
		being no class-$J$ job is the same as that of a reduced system where there
		are no class-$J$ jobs, $\pi^C_{-J}$.
		
		\item[(ii)] The distribution of the state of the system ahead of the first
		class-$J$ job given there is one is also $\pi^C_{-J}$.
		
		\item[(iii)] The distribution of the state of the system after the first
		class-$J$ job given there is one is the same as the distribution of a
		multiclass M/M/1 queue with arrival rate $\lambda $ and service rate $\mu $.
		
		\item[(iv)] The number of class-$J$ jobs in the system in steady state, $%
		N^{J}$, satisfies $N^{J}\sim geom(1-\rho _{J})$, i.e., it is the same as in
		an M/M/1 queue with arrival rate $\lambda _{J}$ and service rate $\hat{\mu}%
		_{J}=\mu -(\lambda -\lambda _{J})$.
	\end{description}
\end{theorem}

Let $T^{i}$ be the response time (total time in system) for a class-$i$ job
in steady state for our collaborative model, and let $T^{M/M/1}(\lambda ,\mu
)$ be the steady-state response time of a job in a standard M/M/1 queue with
arrival rate $\lambda $ and service rate $\mu $, i.e., $T^{M/M/1}(\lambda
,\mu )$ is exponentially distributed with rate $\mu -\lambda $ as long as $%
\lambda <\mu $. Let $T_{Q}^{i}$ and $T_{Q}^{M/M/1}(\lambda ,\mu )$ be
similarly defined for steady-state time in queue.

\begin{corollary}
	\label{fully flexible response}For the collaborative model with a fully
	flexible job class $J$,
	
	\begin{description}
		\item[(i)] $\pi^C (\emptyset )=\pi^C _{-J}(\emptyset )(1-\rho _{J}),$
		
		\item[(ii)] $T^{J}\sim T^{M/M/1}(\lambda _{J},\hat{\mu}_{J})\sim
		T^{M/M/1}(\lambda ,\mu )$, and $T_{Q}^{J}\sim T_{Q}^{M/M/1}(\lambda _{J},\hat{\mu%
		}_{J})$.
	\end{description}
\end{corollary}

\begin{proof}
	(i) From (i) and (iv) of Theorem \ref{fully flexible} we have $\pi^C
	_{-J}(\emptyset )=$ $\pi^C (\emptyset )/P\{N^{J}=0\}=\pi^C (\emptyset
	)/(1-\rho _{J})$.
	
	(ii) Distributional Little's law tells us, for any $\lambda _{a}$ and $L$,
	that if the number of jobs in a queueing system is geometrically distributed
	with mean $L$, jobs arrive at rate $\lambda _{a}$, and jobs are served in
	FCFS order, then the the response time is exponentially distributed with
	mean $L/\lambda _{a}$. The result follows from (iv) of Theorem \ref{fully
		flexible} with arrival rate $\lambda _{a}=\lambda _{J}$ and mean number in
	system $L=\frac{1-\rho _{J}}{\rho _{J}}=\frac{\lambda _{J}}{(\mu -(\lambda
		-\lambda _{J}))-\lambda _{J}}=\frac{\lambda _{J}}{\mu -\lambda }$. Thus, the
	queueing system for class-$J$ jobs in steady state is stochastically
	indistinguishable from a single-class M/M/1 queue with only class-$J$
	jobs and with effective service rate $\hat{\mu} = \mu -(\lambda -\lambda _{J})$.
\end{proof}

Our results for a fully flexible class in the collaborative model can be
extended to general OI queues. Suppose we have an OI queue, so the service
rate as a function of the ordered list of job classes, $\mu (\vec{c}_{n})$,
satisfies conditions (i)-(iii) of Section \ref{OI}, and suppose there is a
maximal service rate $\mu $, such that $\mu (\vec{c}_{n})\leq \mu $ for any
state $\vec{c}_{n}$. Also suppose there is a job class $J$ such that for any
state $\vec{c}_{n}$ in which the first class-$J$ job is in position $k$, $%
k\leq n$, $\mu (\vec{c}_{n})=\mu (\vec{c}_{k})=\mu $. Then a class-$J$ job
will \textquotedblleft block\textquotedblright\ jobs behind it in the OI
queue in the same way a fully flexible job blocks jobs behind it in the
skill-based collaborative queue, and Theorem \ref{fully flexible} and
Corollary \ref{fully flexible response} still hold.

\subsubsection{Other classes in nested systems}

Recall that a nested system has a fully flexible job class, $J$, and if
class $J$ is removed, it decomposes into two or more nonoverlapping nested
subsystems. Thus, each job class $i$, by removing job classes $j$ such that $%
S_{i}\subset S_{j}$ or $S_{i}\cap S_{j}=\emptyset $, defines a nested
subsystem with servers $S_{i}$ and job classes $j$ that {\em require}
servers $S_{j}\subseteq S_{i}$, and where class $i$ is fully flexible. That is, for a subset $S$ of servers, let $%
R(S)=\overline{C(\overline{S})}=\{1,\ldots ,N\}\backslash C(\{1,\ldots
,M\}\backslash S)$ be the job classes that require (i.e., that are only compatible with) servers in $S$.
The nested subsystem defined by job class $i$
consists of servers $k\in S_{i}$ and job classes $j\in R(S_{i})$, i.e., the reduced system
$\vdash \{1,...,M\}\backslash{S_{i}}$.
Let $\hat{\mu}_{i}=\mu (S_{i})-\lambda (R(S_{i}))+\lambda _{i}$ be the effective service capacity for class $i$ in this subsystem, and let $\rho _{i}=\lambda _{i}/\hat{\mu}_{i}$. 
We will show that the overall response time for class-$i$ jobs is the sum of the queueing times for classes $j$ with $S_i \subset S_j$, plus the response time
for class $i$ given those classes are gone (so it is the most flexible class
in its subsystem). Note that, as we observed for class $J$, $%
T^{M/M/1}(\lambda _{i},\hat{\mu}_{i})\sim T^{M/M/1}(\lambda (R(S_{i})),\mu
(S_{i}))$.

\begin{theorem}
	\label{nested}In a nested collaborative system, for any job class $i$,
	\[
	T^{i}\sim T^{M/M/1}(\lambda _{i},\hat{\mu}_{i})+\sum_{j:S_{i}\subset
		S_{j}}T_{Q}^{M/M/1}(\lambda _{j},\hat{\mu}_{j}), 
	\]%
	where all the terms are independent. Also, $\pi^C(\emptyset)=\prod\limits_{j=1}^{J}(1-\rho _{j})$.
\end{theorem}

\begin{proof}
	We start with the response time result. Let class $G$ be fully redundant in
	one of the subsystems obtained when class $J$ is removed. That is, $G$ is
	such that $S_{j}\subset S_{G}$ or $S_{j}\cap S_{G}=\emptyset $ for all $%
	j\neq G,J$. We will show that $T^{G}\sim T_{Q}^{J}(\lambda _{J},\hat{\mu}%
	_{J})+T^{M/M/1}(\lambda _{G},\hat{\mu}_{G})$.
	The result will follow by repeating the argument.
	
	From PASTA and (iv) of Theorem \ref{fully flexible}, an arriving (tagged)
	class-$G$ job in steady state will \textquotedblleft see\textquotedblright\ $%
	N^{J}$ class-$J$ jobs in the system, and it will not be able to start
	service until all of those $N^{J}$ class-$J$ jobs have left the system. That
	is, if there are class-$J$ jobs in the system, the tagged job must wait
	until the end of a class-$J$ busy period, which, for an M/M/1 queue, is the
	same as the class-$J$ response time. Thus, the time the tagged job must wait
	until the system is empty of class-$J$ jobs is%
	\[
	I\{N^{J}>0\}T^{J}\sim T_{Q}^{J}\sim T_{Q}^{M/M/1}(\lambda _{J},\hat{\mu}_{J})%
	\text{.} 
	\]
	
	If $N^{J}=0$ when the tagged class-$G$ job arrives, then from (i) of Theorem \ref%
	{fully flexible}, it will \textquotedblleft see\textquotedblright\ the
	reduced system in steady state, with distribution $\pi^C _{-J}$. If $N^{J}>0$,
	then, from quasi-reversibility, the state left behind by a class-$J$ job
	will have the same distribution as that seen upon arrival. Therefore, given
	it is the last class-$J$ job, i.e., it leaves behind no class-$J$ jobs, then
	the state it leaves behind has the distribution $\pi^C _{-J}$, again from (i)
	of Theorem \ref{fully flexible}. Thus, once there are no class-$J$ jobs, the
	tagged job sees independent subsystems defined by the fully flexible class
	in each. The subsystems that do not include class $G$ will have no effect on
	our tagged job. Hence, applying Corollary \ref{fully flexible response} to
	the subsystem with $G$ instead of $J$ as the most flexible class, we have
	that the class-$G$ response time given there are no class-$J$ jobs is $%
	T^{G}|N^{J}=0$ $\sim T^{M/M/1}(\lambda _{G},\hat{\mu}_{G})$, and the overall
	response time result follows.
	
	From our earlier observations, $\pi^C _{-J}(\emptyset )=$ $\pi^C (\emptyset
	)/P\{\vec{C}\in {\cal C}_{-J}\}=\pi^C (\emptyset )/P\{N^{J}=0\}=\pi^C
	(\emptyset )/(1-\rho _{J})$, so $\pi^C (\emptyset )=(1-\rho _{J})\pi^C
	_{-J}(\emptyset )$. If there are no class $J$ jobs, the system decomposes
	into $K$ independent subsystems, each with its own fully flexible class, $%
	G_{k}$, $k=1,...,K$, so%
	\[
	\pi^C _{-J}(\emptyset )=\prod\limits_{k=1}^{K}\pi^C _{-(J,G_{k})}(\emptyset
	)=\prod\limits_{k=1}^{K}(1-\rho _{G_{k}})\pi^C _{-J}(\emptyset ).
	\]%
	Repeating the argument within each subsystem we get $\pi^C
	(\emptyset)=\prod\limits_{i=1}^{J}(1-\rho _{i})$.
\end{proof}

We have already established that the effective service time of the fully
flexible class $J$, $S_{eff}^{J}=T^{J}-T_{Q}^{J}$, is exponentially
distributed with rate $\hat{\mu}_{J}=\mu -(\lambda -\lambda _{J})$. We can
also see this from our result above. Define the effective service time of a
(tagged) class-$J$ job as the time from which it first has no class-$J$ jobs
ahead of it until it completes service. At the time this effective service
period starts, the system the tagged job sees will decompose into $K$
independent subsystems, each with its own fully flexible class, $G_{k}$, $%
k=1,...,K$, and the tagged job will join each of those subsystems as a fully
flexible job for the subsystem (viewing the collaborative model as a cancel-on-completion redundancy system). From
Corollary \ref{fully flexible response} applied to $G_{k}$ in subsystem $k$,
the response time of the fully flexible class within the subsystem will have
the same distribution as the response time in the corresponding M/M/1 queue, so 
\[
S_{eff}^{J}=\min_{k=1,...,K}\{T^{M/M/1}(\lambda _{k},\mu (S_{k})-(\lambda
(R(S_{k}))-\lambda _{k})\}\sim \Exp (\mu (S_{k})-\lambda
(R(S_{k}))\ \sim \Exp (\mu -(\lambda -\lambda _{J})),
\]%
using the fact that the minimum of exponentials is exponential with the sum
of the rates.

As an example, consider the W model in which class-$i$ jobs can only be
served by server $i$, $i=1,2$, and class-3 jobs can be served by either
server. Then 
\[
T^{3}\sim T^{M/M/1}(\lambda _{3},\mu -\lambda _{1}-\lambda _{2})\text{ and }%
T^{i}\sim T_{Q}^{M/M/1}(\lambda _{3},\mu -\lambda _{1}-\lambda
_{2})+T(\lambda _{i},\mu _{i})\text{, }i=1,2.
\]

\subsection{Noncollaborative Model}
\label{sec:nested_noncollab}

Let $T_{Q|B}^{i}$ be the stationary time in the job queue for a class-$i$ job in the noncollaborative model, given that the set of busy servers is $B=\{1,...,M\}$ (i.e., all servers are busy). 
Then, from Observation~\ref{obs:samplepaths}, we know $T_{Q|B}^{i}$ has the same distribution as the response time for class-$i$ jobs in the collaborative model. Therefore, from Theorem \ref{nested}, we have 

\begin{theorem}
	In a nested noncollaborative system, for any job class $i$, given busy
	servers $B=\{1,...,M\}$,  
	\[
	T_{Q|B}^{i}\sim T^{M/M/1}(\lambda _{i},\hat{\mu}_{i})+\sum_{j:S_{i}\subset
		S_{j}}T_{Q}^{M/M/1}(\lambda _{j},\hat{\mu}_{j}),
	\]%
	where $\hat{\mu}_{j}=\mu (S_{j})-\lambda (R(S_{j}))+\lambda _{j}$ and all
	the terms are independent.
\end{theorem}

The result can be generalized for class $i$, if some servers are idle but
all the servers in $S_{i}$ are busy, as follows. Fix $i$ and the set of busy
servers $B\supseteq S_{i}$, and let $y$ be such that $S_{i}\subseteq
S_{y}\subseteq B$ and $\nexists j\neq y$ such that $S_{y}\subset
S_{j}\subseteq B$. That is, class $y$ determines a nested subsystem
of busy servers in which class $y$ is fully flexible, and there are no jobs of class $j$ such that $S_{y}\subset S_{j}$ in the job queue. 
Therefore, an arriving class-$i$ job sees a reduced system, $\vdash \{1,...,M\}\backslash S_{y}$, consisting only of the servers in $S_{y}$ and job classes $j\in R(S_{y})$, and in which all the servers in $S_{y}$ are busy. We have the following.

\begin{corollary}
	In a nested noncollaborative system, for any job class $i$, given the
	servers in $B\supseteq S_{i}$ are busy, 
	\[
	T_{Q|B}^{i}\sim T^{M/M/1}(\lambda _{i},\hat{\mu}_{i})+\sum_{j:S_{i}\subset
		S_{j}\subseteq S_{y}}T_{Q}^{M/M/1}(\lambda _{j},\hat{\mu}_{j}),
	\]%
	where $\hat{\mu}_{j}=\mu (S_{j})-\lambda (R(S_{j}))+\lambda _{j}$ and all
	the terms are independent.
\end{corollary}

We can use our results for queueing times to obtain response time
distributions for the special case in which the service rate is the same at
all servers; that is, $\mu _{j}=\mu /M$ for all servers $j=1,\dots ,M$. 
We do this by conditioning on the set of busy servers seen by an arriving
job. Define $I_{i}$ as an indicator that the all the servers in $S_{i}$ are busy (other servers may also be busy). 
If an arriving class-$i$ job finds an idle compatible server, it will
immediately enter service; otherwise it must wait in the job queue before
entering service. Hence we have the following, where $T^{i}$ is the class-$i$
response time.

\begin{corollary}
	In a nested noncollaborative system, for any job class $i$, given the
	servers in $B\supseteq S_{i}$ are busy, 
	\[
	T^{i}\sim \Exp(\mu /M)+I_{i} \left( T^{M/M/1}(\lambda _{i},\hat{\mu}%
	_{i})+\sum_{j:S_{i}\subset S_{j}\subseteq S_{y}}T_{Q}^{M/M/1}(\lambda _{j},%
	\hat{\mu}_{j}) \right),
	\]%
	where $\hat{\mu}_{j}=\mu |S_{j}|/M-\lambda (R(S_{j}))+\lambda _{j}$ and all
	the terms are independent.
\end{corollary}

\subsubsection{Challenges in generalizing the NC model}

Unlike for the collaborative model, the results in the previous section
require $\mu _{j}=\mu /M$ for all servers $j=1,\dots ,M$. This condition
ensures that a job's service time is the same regardless of the server on
which it runs. If we were instead to allow different servers to have
different rates, the analysis would change in several ways. First, for a job
that finds multiple compatible servers idle, we need to further condition on
the server on which the job runs. Under RAIS this is determined
probabilistically according to the assignment rule; the probabilities can be
determined using the process described in~\cite{VAW}. Under ALIS this is
determined by which server has been busy longer. 
Second, for a job that finds all compatible servers
busy we still need to determine the server on which the job ultimately runs.
With the current approach, we would need to determine the probability that a
class-$i$ job completes on server $j$ in the collaborative model; this would
then be equal to the probability that a class-$J$ job runs on server $i$ in
the noncollaborative model. Unfortunately, computing this quantity appears
to be complicated.

A final challenge in the NC model is that it is difficult, in general, to
compute the probability that various subsets of servers are busy. While this
analysis is tractable in certain small nested systems, for example, in the W
model, the form of these probabilities is not particularly clean or
intuitive. 
In larger nested systems, we believe that the probabilities needed to
perform the requisite conditioning are unlikely to have a clean closed form
solution, even for a symmetric nested system.

\section{Partial State Aggregation and Conditional Queueing Times}
\label{sec:partial_agg}

Section~\ref{sec:nested} provides one approach for understanding the form of the per-class response time distributions in nested systems.
In this section, we turn to a second approach that uses an alternative, partially aggregated, state description, which
gives us conditional queueing times, given the busy servers in the
order of the jobs they are serving,
for general, possibly non-nested, systems.
Like the detailed states considered in Section~\ref{sec:detailed_states}, the partially aggregated states also provide a Markov description for the model and also yield a product form stationary distribution.

\subsection{Noncollaborative Model}

Instead of tracking the classes of all jobs in the system, we now track the number of jobs in the queue in between jobs in service, but not their individual classes.
Let $l$ denote the number of jobs currently in service.
The partially aggregated state includes the vector $\vec{n}_l = (n_1,\dots,n_l)$, where $n_i$ denotes the number of jobs in the queue (\emph{not} in service) that arrived after the $i$th job in service and before the $(i+1)$st job in service.
Under both ALIS and RAIS, we track the busy servers in the arrival order of the jobs they are serving, (but \emph{not} the classes of the jobs in service); for the ALIS version we also track the idle servers in the order in which they became idle.

The partially aggregated state description for  noncollaborative models was first introduced by Adan, Visschers, and Weiss~\cite{AW,VAW}. In these papers, the stationary distribution was derived directly using partial balance for the partially aggregated states.
In this section we provide an alternative derivation that involves aggregating the stationary probabilities for the detailed states discussed in Section~\ref{sec:detailed_states}.

Let us first consider the noncollaborative model with the RAIS policy, in which arrivals finding multiple idle compatible servers are assigned to a server at random with appropriate probabilities that depend on the set of busy servers.
The partially aggregated state is $(\vec{b}_{l},\vec{n}_{l})$, where $l$ is the number of busy servers (which in the noncollaborative model is the same as the number of jobs in service), $\vec{b}_{l}$ is the set of busy servers in order of the arrival times of the jobs they are serving, and $n_{i}$, $i=1,...,l$ is the number of jobs waiting for one of busy servers $1,\ldots ,l$. 
Thus, server $b_{1}$ is serving the oldest job, the next $n_{1}$ jobs to have arrived are waiting for (require) server $b_{1}$, i.e., their classes are in $R(b_{1})$, the $n_{1}+2^{nd}$ oldest job is being served by $b_{2}$, the next $n_2$ jobs are only compatible with $b_1$ or $b_2$ or both, i.e. their classes are in $R(\vec{b}_{2})$, and so on. The corresponding detailed state, $\vec{z}_m$, is such that $z_1 = b_1$, $z_{n_{1}+2}=b_2$, etc., and $m=l+\sum_{i=1}^{l}n_{i}$.
Thus, in the partially aggregated state $(\vec{b}_{l},\vec{n_{l}})$ there are $l$ jobs in service and $\sum_{i=1}^{l}n_{i}$ jobs in the queue. 
When the set of busy servers is $\vec{b}_{j}$, let $\lambda ^{a}_{s}(\vec{b}_{j})$ represent the activation rate of idle server $s\notin \{b_{1},...,b_{j}\}$ (the rate of going from state $(\vec{b}_{j},\vec{n}_{j})$ to $((\vec{b}_{j},s),(\vec{n}_{j},0)$).

With this description of the state we defer determining the class of a job until we need it. 
That is, we realize information about a job's class only when a server becomes available and the job under consideration is next in the queue behind the available server.
At this point, we probabilistically determine whether or not the job is compatible with the server; if it is compatible, it enters service.
If not, the server ``skips over'' the job and we have narrowed down the set of possible classes for the job, but we may have not specified its exact class.

Visschers et al.\ found that the above state space exhibits a product form stationary distribution, under the \emph{assignment condition} for routing a compatible job to 
idle server $b_j$ given busy servers $\vec{b}_{j-1}$: $\prod\limits_{j=1}^{l}\lambda ^{a}  _{b_{j}}(\vec{b}%
_{j-1})$ must be the same for any permutation of $b_{1},\ldots ,b_{l}$.
This is the same assignment condition that we use in Section~\ref{sec:detailed_states} for the detailed state description.

Let $\alpha (\vec{b}_{j})=\frac{\lambda (R(\vec{b}_{j}))}{\mu (\vec{b}_{j})}$. We use the notation $\pi^{RAIS'}$ to denote the partially aggregated stationary distribution under RAIS (in contrast with $\pi^{RAIS}$, which denotes the stationary distribution for the detailed state description).

\begin{theorem}
	\label{thm:partialagg_rais}
	(Visschers et al.~\cite{VAW})
	\begin{align*}
	\pi^{RAIS'} (\vec{b}_{l},\vec{n}_{l}) &= \pi^{RAIS}(\emptyset)
	\prod_{j=1}^l \frac{\lambda^{a} _{b_{j}}(\vec{b}_{j-1})}{\mu
		(\vec{b}_{j})} 
	\alpha (\vec{b}_{j})^{n_{j}} 
	=\pi^{RAIS}(\emptyset)  \prod\limits_{j=1}^{l}\alpha (\vec{b}%
	_{j})^{n_{j}}\prod\limits_{j=1}^{l}\frac{\lambda ^{a}_{b_{j}}(\vec{b}_{j-1})}{%
		\mu (\vec{b}_{j})} \\
	&=\pi^{RAIS} (\vec{n}_{l}\vec{|b}_{l})\pi^{RAIS} (\vec{b}_{l})=\prod\limits_{j=1}^{l}(1-%
	\alpha (\vec{b}_{j}))\alpha (\vec{b}_{j})^{n_{j}}\pi^{RAIS} (\vec{b}_{l}).
	\end{align*}
\end{theorem}

The proof given by Visschers et al.\ involves showing directly that local balance holds for the partially aggregated states.
Below we give an alternative proof, which follows by summing the stationary probabilities (given in Theorem~\ref{thm:noncollab_rais_detailed}) of states $\vec{z}_{m}$ that are consistent with $(\vec{b}_{l},\vec{n_{l}})$.

\begin{proof}
	We begin by recalling that $\vec{z}_m$ is an interleaving of states $\vec{c}_n$ and $\vec{b}_l$, where $m = n + l$.
	That is, letting $k_i = \sum_{j=1}^i n_j$, we can write
	$$\vec{z}_m =  (b_1,c_1,\dots,c_{n_1},b_2,\dots,b_j,c_{k_j +1} ,\dots,c_{k_{j-1}+n_{j}},b_{j+1},\dots,b_l,c_{k_{l-1} +1} ,\dots,c_{k_{l-1}+n_{l}}).$$
	
	Let ${\cal C}(\vec{b}_l,\vec{n}_l)$ denote the set of states $\vec{z_m}$ compatible with $(\vec{b}_l,\vec{n}_l)$.
	We then have
	\begin{align*}
	\pi^{RAIS'}(\vec{b}_l,\vec{n}_l) &= \sum_{\vec{z}_m \in C(\vec{b}_l,\vec{n}_l)} \pi^{RAIS}(\vec{z}_m) \\
	&= \pi^{RAIS}(\emptyset) \sum_{\vec{z}_m \in C(\vec{b}_l,\vec{n}_l)} \prod_{i=1}^{m} \frac{\lambda_i^z(\vec{z}_i)}{\mu(\vec{z}_i)} \\
	&= \pi^{RAIS}(\emptyset) \sum_{\vec{z}_m \in C(\vec{b}_l,\vec{n}_l)} \prod_{j=1}^{l} \left( \frac{\lambda_{b_j}^a(\vec{b}_{j-1})}{\mu(\vec{b}_j)} 
	\prod_{i = k_{j-1} +1}^{k_{j-1}+n_j}
	\frac{\lambda_{c_i}}{\mu(\vec{b}_j)}  \right)  \\
	&= \pi^{RAIS}(\emptyset) \prod_{j=1}^{l} \left(  \frac{\lambda_{b_j}^{a}(\vec{b}_{j-1})}{\mu(\vec{b}_j)} 
	\prod_{i = k_{j-1} +1}^{k_{j-1}+n_j}
	\frac{\sum_{c \in R(\vec{b}_j)} \lambda_c }{\mu(\vec{b}_j)} \right)\\
	&= \pi^{RAIS}(\emptyset) \prod_{j=1}^{l} \left(  \frac{\lambda_{b_j}^{a}(\vec{b}_{j-1})}{\mu(\vec{b}_j)} \left( \frac{\lambda(R(\vec{b}_j)) }{\mu(\vec{b}_j)} \right)^{n_j} \right)\\
	&= \pi^{RAIS}(\emptyset) \prod\limits_{j=1}^{l}\alpha (\vec{b}%
	_{j})^{n_{j}}\prod\limits_{j=1}^{l}\frac{\lambda _{b_{j}^{a}}(\vec{b}_{j-1})}{%
		\mu (\vec{b}_{j})}.
	\end{align*}
\end{proof}

We now turn to the ALIS policy, in which arrivals finding multiple idle compatible servers are assigned to the one that has been idle longest. 
Because the order of the idle servers now affects the system evolution, we now consider the aggregate state $(\vec{s}_{M-l},\vec{b}_{l},\vec{%
	n_{l}})$, where $\vec{s}_{M-l}$ is the set of $M-l$ idle servers in the
order in which they became idle,
and $\vec{b}_l$ and $\vec{n}_l$ are defined as in the RAIS model.

\begin{theorem} (Adan and Weiss~\cite{AW})
	\[
	\pi^{ALIS'}(\vec{s}_{M-l},\vec{b}_{l},\vec{n_{l}})= \pi^{ALIS}(\emptyset,\emptyset) \prod\limits_{j=1}^{l}%
	\alpha (R(\vec{b}_{j}))^{n_{j}}%
	\prod\limits_{j=1}^{l}\frac{1}{\mu (\vec{b}_{j})}
	\prod\limits_{j=1}^{M-l}\frac{1}{\lambda (\vec{s}_{j})} 
	\]%
	where $\pi^{ALIS}(\emptyset,\emptyset)$ is a normalizing constant. 
\end{theorem}

As noted by Adan and Weiss~\cite{AW}, aggregating the stationary distribution under ALIS over all permutations of the idle servers, $\vec{s}_k$, yields the same stationary distribution as under RAIS.

\begin{corollary}
	\label{cor:agg_alis}
	$\sum_{{\cal{P}}(\vec{s}_{M-l})} \pi^{ALIS'}(\vec{s}_{M-l},\vec{b}_l,\vec{n}_l) = \pi^{RAIS'}(\vec{b}_l,\vec{n}_l)$.
\end{corollary}

Corollary~\ref{cor:agg_alis} tells us that the conditional stationary distribution of the time in queue, given the set of busy servers, is the same under ALIS as under RAIS.
Under both policies, conditioned on $\vec{b}_l$, the number of jobs waiting in the queue between busy servers $b_j$ and $b_{j+1}$ is geometrically distributed with parameter $1 - \alpha(R(\vec{b}_j)) = 1 - \lambda(R(\vec{b}_j))/\mu(\vec{b}_j)$, from Theorem~\ref{thm:partialagg_rais}. Moreover, each of these jobs is 
of class $c$ with probability $\lambda_c/\lambda(R(\vec{b}_j))$. 
Therefore, from Lemma~\ref{split}, conditioned on $\vec{b}_l$, the number of class-$i$ jobs waiting in the queue between busy servers $b_j$ and $b_{j+1}$ is geometrically distributed with parameter $1 - \frac{\lambda_i}{\mu(\vec{b}_j)-\lambda(R(\vec{b}_j))+\lambda_i}$.
Hence, from distributional Little's law, and because class-$i$ jobs are
served in order, the time a class-$i$ job will spend in the ``subsystem'' behind the servers $\vec{b}_j$ is the same
as the response time for a standard M/M/1 queue with arrival rate 
$\lambda_i$ and service rate $\mu(\vec{b}_j)-\lambda(R(\vec{b}_j))+\lambda_i$. Therefore, depending
on $\vec{b}_l$,
a job arriving in steady state will either start service immediately, or wait a sum of exponential times before entering service.

\begin{theorem}
	The queueing time for a class $i$ job, given busy servers $\vec{b}_{l}$, is%
	\[
	T^i_{Q}(\vec{b}_{l})=I(i\in R(\vec{b}_{l})) \sum_{j=f(i,\vec{b}%
		_{l})}^{l}T^{M/M/1}(\lambda _{i},\mu (\vec{b}_{j})-\lambda (R(\vec{b}_{j}))+\lambda
	_{i}) 
	\]%
	where $I(\cdot)$ is the indicator function, and $f(i,\vec{b}_{l})=\arg
	\min \{j:0\leq j\leq l,i\in R(\vec{b}_{j})\}=\arg \max \{j:0\leq j\leq
	l,b_{j}\in S(i)\}$ is the largest indexed busy machine that is compatible
	with job class $i$.
\end{theorem}

We note that, unlike the results for nested systems derived in Section~\ref{sec:nested}, here the per-class queueing time distribution is conditioned on the ordered vector of busy servers.
In general it is not straightforward to obtain closed-form expressions for the probability that a job sees a particular $\vec{b}_l$.
Hence, while this form is insightful in terms of interpreting the time in queue as that in a tandem series of M/M/1 queues, the form does not permit an easy derivation of mean time in queue or other exact performance metrics.

\subsection{Collaborative Model}

In this section we note briefly that a similar partially aggregated state can be defined for the collaborative model. Here the partially aggregated state description is $(\vec{d}_{l},\vec{n}_{l})$, where $\vec{d}_l$ gives the classes of all jobs currently in service, and
$\vec{n}_l$ gives the number of jobs in the queue (receiving no service) in between those jobs in service.
This is similar to the $(\vec{b}_l,\vec{n}_l)$ state used for RAIS, except that now we track the classes of the job in service rather than the servers processing these jobs. We define $\mu(\vec{d}_i)$ as the total service rate given to the first $i$ jobs that are receiving service, and $R(\vec{d}_i)$ as the classes of jobs that require one of the servers serving the jobs in $\vec{d}_i$. That is, $c\in R(\vec{d}_i)$ if $S_c \in S(\vec{d}_i)$.
Note that for $d_i$ to be in service, given $\vec{d}_{i-1}$, we
must have $d_i \notin \vec{d}_{i-1}$. Let 
$\alpha(\vec{d}_i) = \frac{\lambda(R(\vec{d}_i))}{\mu(\vec{d}_i) }$.

\begin{proposition}
	\label{dn pi}For $l=0,\ldots ,M$, $\vec{d}_{l}$ such that 
	$d_i \notin \vec{d}_{i-1}$ for $i=2,\ldots ,l$, and $%
	n_{i}=0,1,\ldots $ for $i=1,\ldots ,l$, 
	\[
	\pi^{C'} (\vec{d}_{l},\vec{n}_{l})=\pi^C \left( \emptyset\right) \prod\limits_{j=1}^{l}%
	\frac{\lambda _{d_{j}}}{\mu (\vec{d}_{j})}\alpha (\vec{d}_{j})^{n_{j}}. 
	\]
\end{proposition}

The proof follows a similar argument to the proof of Theorem~\ref{thm:partialagg_rais} by aggregating detailed states consistent with $(\vec{d}_l,\vec{n}_l)$; we omit the details. 

As for the noncollaborative system, we can use Proposition~\ref{dn pi} and the distributional form of Little's Law to determine the distribution of $T_Q^i(\vec{d}_l)$, the conditional waiting time until 
a job starts service on at least one server for any class $i$, given $\vec{d}_{l}$.

\begin{corollary}
	\[
	T^i_{Q}(\vec{d}_{l})=I(i\in R(\vec{d}_{l})) \sum_{j=f(i,\vec{d}%
		_{l})}^{l}T^{M/M/1}(\lambda _{i},\mu (\vec{d}_{j})-\lambda (R(\vec{d}_{j}))+\lambda
	_{i}) 
	\]%
	where $I(\cdot)$ is the indicator function, and $f(i,\vec{d}_{l})=\arg
	\min \{j:0\leq j\leq l,i\in R(\vec{d}_{j})\}=\arg \max \{j:0\leq j\leq
	l,S(d_{j}) \cup S(i) \neq \emptyset$ is the largest indexed job in service that is using a server in $S_i$.
\end{corollary}

\begin{corollary}
	$T_Q^i \sim \sum_{l} \sum_{\vec{d}_l} T_Q^i(\vec{d}_l) I(\vec{d}_l)$.
\end{corollary}

\section{Per-class State Aggregation and Mean Performance Measures}
\label{sec:per_class}

In this section we consider class-based performance measures for the OI
queue (and hence the collaborative model, and for the queue of the noncollaborative model given all servers busy) with general, non-nested, bipartite
structures. Here it is useful to define the per-class aggregated state $%
x=(x_{1},...,x_{N})$, where $x_{i}$ is the total number of type $i$ jobs.
Let $C(x)$ be set of states of the form $\vec{c}_{n}$ that are consistent
with $x$, i.e., $\vec{c}_{n}\in C(x)$ if and only if $x_{i}=\sum_{j=1}^{n}I%
\{c_{j}=i\}$ for all classes $i$. so $n=\sum x_{i}$. Abusing notation, let $%
\mu (x)$ be the total service rate in state $x$; from OI property (ii) (see
Section~\ref{sec:detailed_states}), $\mu (x)=\mu (\vec{c}_{n})$ for all
states $\vec{c}_{n}\in C(x)$. Then $\pi ^{X}(x)=\sum_{\vec{c}_{n}\in
	C(x)}\pi (\vec{c})$, where $\pi ^{X}(x)$ is the steady-state probability of
aggregate state $x$ and, by aggregating the partial balance equations, is given by
\[
\mu (x)\pi ^{X}(x)=\sum_{i:x_{i}>0}\lambda _{i}\pi ^{X}(x-e_{i}),
\]%
where $e_{i}$ is a vector of appropriate length containing a 1 in position $i
$ and 0's elsewhere. Then $\pi ^{X} (x)$ also has a product form. The result
follows from summing the product form characterizations of $\pi ^{X} (\vec{c})$.

\begin{theorem}
	\label{aggregate pi}(Bonald and Comte \cite{BC}, Krzesinski \cite{K})%
	\[
	\pi ^{X} (x)=\pi ^{C}_X(\emptyset )\Phi (x)\prod\limits_{i =1}^{N}\lambda
	_{i}^{x_{i}} 
	\]%
	where 
	\[
	\Phi (x)=\frac{1}{\mu (x)}\sum_{i:x_{i}>0}\Phi (x-e_{i})\text{, }\Phi
	(\emptyset )=1
	\]
	and $\pi^X(\emptyset) = \pi^C(\emptyset)$ is a normalizing constant equal to the probability that the system is empty.
\end{theorem}

Note that the aggregate state description does not capture the dynamics of
the OI queue and is not a Markov description for the original system. While
the order of the jobs, given $x$, does not matter for the {\em total}
service rate, $\mu (x)$, it does matter for the amount of service received
by the $j$'th job in the queue, $\Delta _{j}(\vec{c}_{j})=\Delta _{j}(\vec{c}%
_{n})$. We therefore need to know $\vec{c}_{n}$---or at least the class of
the job in service on each server---to know the rate out of the state due to
a class-$i$ departure.

As observed by Bonald and Comte, the stationary aggregate distribution $\pi^C
_{X}(x)$ given in Theorem \ref{aggregate pi} is also the stationary
distribution of a single-server system consisting of $N$ job classes with
Poisson arrivals at rates $\lambda _{i}$ and state-dependent exponential
service rates such that class-$i$ jobs are served according to processor
sharing at rate 
\[
\phi _{i}(x)=[\Phi (x-e_{i})/\Phi (x)]I(x_{i}>0).
\]%
We call the single-server model with service rates $\phi _{i}(x)$ the {\em %
	aggregate model}. Bonald and Comte also noted the following relationship
between the aggregate model and the collaborative model: 
\[
\phi _{i}(x)=\sum_{\vec{c}_{n}\in C(x)}\frac{\pi^{C} (\vec{c}_{n})}{\pi ^X(x)}%
\mu _{i}^{\prime }(\vec{c}_{n}),
\]%
where $\mu _{i}^{\prime }(\vec{c}_{n})$ is the service rate of the first
class-$i$ job in state $\vec{c}_{n}$ in the collaborative model. (Because of
the FCFS service discipline, no other class-$i$ jobs will be in service.)
Note that from Theorem \ref{aggregate pi}, $\sum_{i}\phi _{i}(x)=\mu (x)$.

The service rates $\phi_i(x)$ also satisfy the following {\em balance
	property}:%
\[
\phi _{i}(x)\phi _{j}(x-e_{i})=\phi _{j}(x)\phi _{i}(x-e_{j}) 
\]%
for $x_i,x_j > 0$. The balance property is analogous to the assignment
condition required for the noncollaborative model with random assignment to
idle servers to have a product-form stationary distribution. Furthermore,
the balance property leads to Kolmogorov's criterion being satisfied, and
therefore it leads to the system being reversible. 
Because the aggregate model is reversible, it is also insensitive to the job
size distributions.

While the stationary distribution of the states $x$ is the same in the
aggregate model and the collaborative model, as noted above, the underlying system
dynamics are very different. Bonald and Comte propose applying a round robin-like
scheduling algorithm to the original collaborative model to
approximate the behavior of the aggregate model~\cite{BC}. Under their
algorithm, server $j$ serves the first compatible job in the queue, as in
the original, collaborative, FCFS model, but, after an exponential time with
rate $\theta _{j}$, server $j$ interrupts the job in service and that job is
moved to the back of the queue. This is analogous to approximating processor
sharing with round robin for a single server and job class. Bonald and Comte
note that the aggregate model, using balanced fair processor sharing, is insensitive in
that the steady-state distribution does not depend on the job size
distribution.

Because the aggregate model and the collaborative model have the same steady-state
distribution for the aggregate states, we can use the aggregate model to
efficiently compute aggregate performance measures for the collaborative
model. In particular, Bonald et al.~\cite{BCM} give a recursion based on successively
removing servers for computing the system idle probability, $\pi^{C} (\emptyset
)=\pi ^{X}(\emptyset )$, as follows.

Let ${\cal C}$ be the set of all (detailed) states $\vec{c}_{n}$ for the
original collaborative model. Recall that the subscript $\vdash k$ represents a
reduced system without server $k$, i.e., in which server $k$ as well as the
job classes in $C_{k}$ are removed. Let $\psi _{k}$ be the probability that
server $k$ is idle in the original collaborative system. Then, from Corollary \ref{reduced}, we have that $\pi ^{C} (\vec{c}_{n}|$ server $k$ is idle$%
)=\pi ^{C} _{\vdash k}(\vec{c}_{n})$ for $\vec{c}_{n}\in {\cal C}_{\vdash k}$, so $\pi ^{C} (%
\vec{c}_{n})=$ $\pi ^{C} _{\vdash k}(\vec{c}_{n})\psi _{k}$. Then $\pi ^{C} (\emptyset
)= \pi ^{X} (\emptyset) = \pi ^{C} _{\vdash k}(\emptyset )\psi _{k}$ (and hence $\psi _{k}$) can be computed recursively:

\begin{proposition}
	(Bonald et al. \cite{BCM}) 
	\[
	\pi ^{C} (\emptyset ) = \pi ^{X} (\emptyset )=(1-\rho )\frac{\mu }{\sum_{k=1}^{M}\frac{\mu _{k}}{\pi ^{C}
			_{\vdash k}(\emptyset )}}, 
	\]%
	where $\rho =\lambda /\mu $ is the system load.
\end{proposition}

\begin{proof}
	Algebra, using $\pi ^{C} _{\vdash k}(\emptyset )=\pi ^{C} (\emptyset )/\psi _{k}$, gives us
	that the equation above is equivalent to 
	\[
	\sum_{k=1}^{M}\mu _{k}\psi _{k}=\mu -\lambda \text{.} 
	\]%
	This just represents two ways of computing the long-run rate of
	\textquotedblleft dummy\textquotedblright\ transitions, i.e., potential
	service completions at idle servers.
\end{proof}

Recall that the collaborative model is equivalent to the directed
bipartite matching model, in which servers of type $k$ arrive
according to a Poisson process at rate 
$\mu_k$, and arriving servers that do not find compatible jobs
(unmatched servers)
immediately leave the system. Here the interpretation of $\psi_k$
is the probability that an arriving server of type $k$ is unmatched, and ``dummy'' transitions
correspond to arrivals of unmatched servers. See Weiss~\cite{W} for
an alternative algorithm to compute $\psi_k$ and $\pi^C (\emptyset)$,
as well as for computing the long-run matching rates of class $i$ jobs
with class $k$ servers.

The mean number of jobs $L$ and the mean number of class-$i$ jobs $L_{i}$,
with $L_{\vdash k}$ and $L_{i \vdash k}$ similarly defined for the reduced system
without server $k$ and its compatible job classes, can be similarly
recursively calculated. Note that $L_{\vdash k}$ and $L_{i \vdash k}$ are also the
conditional mean number of class-$i$ jobs in the original system, given
server $k$ is idle. Let $\bar{S}_{i}=\{1,\ldots ,M\}\backslash S_{i}$ denote
the set of servers that cannot serve class-$i$ jobs, and let $\rho
_{i}=\lambda _{i}/(\mu -(\lambda -\lambda _{i}))$ be the mean number of
class-$i$ jobs in an M/M/1 queue with arrival rate $\lambda _{i}$ and
service rate $\mu -(\lambda -\lambda _{i})$.

\begin{proposition}
	\begin{align*}
	L_{i}& =\frac{\lambda _{i}+\sum_{k\in \overline{S}_{i}}\mu _{k}\psi
		_{k}L_{i|-k}}{\mu -\lambda }=\frac{\lambda _{i}}{\mu -\lambda }+\sum_{k\in 
		\overline{S}_{i}}\frac{\mu _{k}\psi _{k}}{\mu -\lambda }L_{i|-k} \\
	L& =\sum_{i=1}^{J}L_{i}=\frac{\lambda +\sum_{k=1}^{M}\mu _{k}\psi _{k}L_{\vdash k}%
	}{\mu -\lambda }=\frac{\lambda }{\mu -\lambda }+\frac{\sum_{k=1}^{M}\mu
		_{k}\psi _{k}L_{\vdash k}}{\sum_{k=1}^{M}\mu _{k}\psi _{k}}.
	\end{align*}
\end{proposition}

Note that $\frac{\lambda _{i}}{\mu -\lambda }$ is the mean number of class-$i
$ jobs in an M/M/1 queue with arrival rate $\lambda _{i}$ and service rate $%
\mu -(\lambda -\lambda _{i})$ (the maximal service rate available to class-$i
$ jobs), as we argued in Section~\ref{sec:nested}. It also represents the mean
number of class-$i$ jobs in the collaborative (not necessarily nested) model
given all the servers are busy. Also, as noted above, $\mu _{k}\psi
_{k}/(\mu -\lambda )$ represents the proportion of dummy transitions due to
server $k$ being idle, so the second set of terms in the above expression
represent the additional expected jobs due to \textquotedblleft
wasted\textquotedblright\ service because of job/server incompatibilities.
Note that servers in $S_{i}$ will not be idle if there are class $i$ jobs in
the collaborative system. 

Bonald et al.\ use the results above to obtain explicit results for special
cases, such as redundancy-$d$, where all jobs are replicated to a randomly
chosen subset of $d$ servers, and line structures in which job classes can
be ordered so that for any server $k$, the classes of jobs it can serve are
consecutive, i.e., classes $i_{1},i_{1}+1,\ldots ,i_{2}$ for some $%
i_{1}<i_{2}$ \cite{BCM}. Nested structures are a special case of line
structures so the above recursions represent an alternative method for
deriving mean performance metrics to the approach given in
Section~\ref{sec:nested}, where, of course, mean response times follow immediately
from Little's Law.

We note that though the results for this section are for the collaborative
model, in light of our observation that the queue process in the
noncollaborative model given all servers are busy has the same distribution
as the overall process for the collaborative model, we can apply the results
above to the noncollaborative case. That is, e.g., $L_{i}$ as computed above
will equal the expected number of class-$i$ jobs in the queue (not receiving
service) in steady state, given all the servers are busy, for the
noncollaborative model.

\section{Related Work}
\label{sec:related_work}

The majority of this paper has focused on surveying results related to product-form stationary distributions and derivations of performance metrics in the collaborative and noncollaborative systems, under a few key assumptions: that service times are exponentially distributed and  i.i.d.\ (across jobs and across replicas of the same job, in the collaborative model), and that the service discipline is FCFS.
There are several lines of work that relax one or more of these assumptions; such relaxations preclude product-form results, and as such fall outside the scope of this paper.
In this section we provide a brief outline of some of the related work.

We begin with related work within the i.i.d.\ exponential model.
Several papers have considered a scheduling policy that gives priority to less flexible jobs over more flexible jobs; this policy is known as ``dedicated customers first'' in the noncollaborative system and as ``least redundant first'' in the collaborative system. Such a policy has been shown to be optimal in the sense that it stochastically maximizes the departure rate, in both the noncollaborative system~\cite{ARW 2013} and the collaborative system~\cite{GHHR}. Furthermore, in the collaborative case mean response time is decreasing and convex under this policy as the proportion of jobs that are more flexible increases~\cite{GHR}. In a similar vein, the effect of increasing the ``degree'' of flexibility (i.e., the number of servers with which each job is compatible) in systems with FCFS scheduling has been studied. In both the noncollaborative~\cite{ARW 2011} and collaborative~\cite{KRW,GHSVZ} systems, mean response time is decreasing and convex as the degree of redundancy increases.
Gurvich and Whitt~\cite{GW1,GW2,GW3} consider other routing and scheduling policies for the noncollaborative model in the many-server heavy-traffic regime.

The system in which all jobs have the same degree of flexibility and are assigned a set of compatible servers uniformly at random is one special case of the system structure considered in this paper. This special case, often referred to as a ``redundancy-$d$'' system (the $d$ indicating the degree of flexibility), has received considerable attention in the literature because the symmetric system structure makes analysis more feasible in many cases. Indeed, the redundancy-$d$ system is another example of a system in which, under the i.i.d\ exponential and FCFS assumptions of this paper, it is feasible to aggregate the product-form stationary distribution to derive performance metrics~\cite{GHSVZ,ABV}. 

Several papers have noted that when the i.i.d.\ exponential assumptions are eliminated, mean response time no longer decreases as $d$ increases in the collaborative model~\cite{GHSV}, hence recently there has been a focus on developing new dispatching and scheduling policies for systems with correlated or general service times~\cite{GHSV,HH,HH2,RBB}. Finally, some work focuses on deriving the stability region under redundancy-$d$, which becomes much more complicated absent the i.i.d.\ exponential assumptions~\cite{AAJV}.

\section{Conclusion}
\label{sec:conclusion}

This paper presents an overview of product-form results in systems with flexible jobs and servers, in which a bipartite graph structure specifies which job classes can be served by which servers.
We primarily focus on two models for service: the collaborative model, in which multiple servers can work together to serve a single job at a faster rate, and the noncollaborative model, in which each job is permitted to enter service on only one server. Both models have been studied extensively in the literature; this survey brings together the two models, as well as several other related systems, using a common language and set of notation. Our hope is that this will allow readers to draw new connections among these similar systems. Along the way, we have presented several new results that highlight the relationships between models and that show how results derived in one model can be used to obtain insights for the other.

One of the primary goals in analyzing queueing systems is to determine response time distributions; in multi-class systems such as those considered in this paper, we wish to derive per-class response time distributions.
Each of the three state descriptors that we consider allows us to make partial progress towards this goal.
Using the detailed state descriptor of Section~\ref{sec:detailed_states}, we can derive per-class response time distributions for the special case of nested systems.
Using the partially aggregated states of Section~\ref{sec:partial_agg}, we can derive per-class queueing time distributions in general (not necessarily nested) systems, but now conditioned on the ordered set of busy servers.
Using the per-class aggregated states of Section~\ref{sec:per_class}, we can derive unconditional per-class mean performance metrics in general systems, but this approach does not yield distributional results.
Each approach has its advantages and disadvantages; we believe that a unifying analysis that provides unconditional per-class response time distributions is likely to be infeasible, but this remains an open question.

\section{Acknowledgements}
We thank Gideon Weiss, Erol Pekoz, and Jan-Pieter Dorsman for their careful reading and valuable feedback.

\end{document}